\documentclass[12pt, draftclsnofoot, onecolumn]{IEEEtran}
\usepackage{amsmath}
\usepackage{amsthm}
\usepackage{exscale,relsize}
\usepackage{cite}
\usepackage{graphicx}
\usepackage{cases}
\usepackage{epstopdf}
\usepackage{amsfonts,amsmath,amssymb}
\usepackage{graphicx}
\usepackage{bm}
\usepackage{bbm}
\usepackage{color}
\usepackage{enumerate}
\DeclareMathOperator*{\st}{s.t.}
\theoremstyle{definition}
\newtheorem{Lemma}{Lemma}

\newtheorem{Assumption}{Assumption}
\newtheorem{prop}{Proposition}

\usepackage{algorithm,algorithmic}

\usepackage[belowskip=-5pt,aboveskip=0pt]{caption}
\allowdisplaybreaks
\ifCLASSINFOpdf
\else
\fi
\begin{document}
%
\title{\LARGE{Energy-Efficient Online Data Sensing and Processing in Wireless Powered Edge Computing Systems}

\author{Xian~Li,~
	Suzhi~Bi,~
	Yuan~Zheng,~
	and~Hui Wang}

\thanks{{X.~Li, S.~Bi and Y. Zheng are with the College of Electronics and Information Engineering, Shenzhen University, Shenzhen 518060, China
		(email: {xianli, bsz, zhyu}@szu.edu.cn). H.~Wang is with the Shenzhen Institute of Information Technology, Shenzhen 518172, China (email: wanghsz@sziit.edu.cn).} }}

%
%
%

\maketitle

\vspace{-4em}
\begin{abstract}
Wireless powered mobile edge computing (MEC) has emerged as a promising paradigm to enable high-performance computation of energy-constrained wireless devices (WDs) in Internet of things (IoT) systems. However, to overcome the severe path loss of both energy transfer and data communications, wireless powered MEC suffers from high operating power consumption. To achieve sustainable and economic system operation, this paper focuses on developing energy-efficient online data processing strategy of wireless powered MEC systems under stochastic fading channels. In particular, we consider a hybrid access point (HAP) transmitting RF energy to and processing the sensing data offloaded from multiple WDs. Under an average power constraint of the HAP, we aim to maximize the long-term average data sensing rate of the WDs while maintaining task data queue stability. We formulate the problem as a multi-stage stochastic optimization to control the energy transfer and task data processing in sequential time slots. Without the knowledge of future channel fading, it is very challenging to determine the sequential control actions that are tightly coupled by the battery and data buffer dynamics. To solve the problem, we propose an online algorithm named LEESE that applies the perturbed Lyapunov optimization technique to decompose the multi-stage stochastic problem into per-slot deterministic optimization problems. We show that each per-slot problem can be equivalently transformed into a convex optimization problem. To facilitate online implementation in large-scale MEC systems, instead of solving the per-slot problem with off-the-shelf convex algorithms, we propose a block coordinate descent (BCD)-based method that produces close-to-optimal solution in less than 0.04\% of the computation delay. Simulation results demonstrate that the proposed LEESE algorithm can provide 21.9\% higher data sensing rate than the representative benchmark methods considered, while incurring sub-millisecond computation delay suitable for real-time control under fading channel. 

\end{abstract}
\vspace{-8pt}
\begin{IEEEkeywords}\vspace{-0em}
	Mobile edge computing, wireless power transfer, computation offloading, resource allocation, online optimization algorithm.
\end{IEEEkeywords}

\section{Introduction}
\subsection{Motivations and Contributions}
As a seamless integration of wireless power transfer (WPT) \cite{Ju2014ThroughputMaximization,Bi2015,Zeng2017,Li2019} and mobile edge computing (MEC) \cite{Mao2017,Wu2018TVT,Bi2021,Wang2016}, wireless powered MEC is recognized as a promising technology to provide sustainable and enhanced computation performance for delay-sensitive and data-intensive IoT (internet of things) applications. With dedicated radio frequency (RF) energy transmitter and MEC server integrated as a hybrid access point (HAP), wireless powered MEC system provides on-demand energy transfer and computation service to remote low-power wireless devices (WDs) (e.g., sensors and wearable devices). Powered by the received energy, WDs collect sensing data and process the data either locally or remotely at the HAP via task offloading. 

A major hurdle to the wide deployment of wireless powered MEC systems is the high power consumption to overcome the double propagation loss during downlink WPT and uplink task offloading. Because the energy consumption of WDs is replenished entirely by WPT in a wireless powered MEC system, it calls for joint optimization of WPT and task execution at both the HAP and WDs to improve the energy efficiency \cite{Wang2018a,Hu2018,Zhou2020}. For example, by jointly optimizing the WPT beamforming and computation resource allocation, \cite{Wang2018a} minimized the total energy consumption of a multi-antenna HAP to complete the task computation of WDs. With the same design objective, \cite{Hu2018} proposed cooperative task offloading in a two-user wireless powered MEC system, and optimized the time and transmit power of WPT and data offloading. Considering both time division multiple access (TDMA) and non-orthogonal multiple access (NOMA), \cite{Zhou2020} maximized the energy efficiency of a wireless powered MEC system to achieve the maximum processed data bits per joule energy consumption.\

Despite such research progress, these prior works \cite{Wang2018a,Hu2018,Zhou2020} focused on independently optimizing the instantaneous system performance within each time slot given the wireless channel gains. {\color{black}In practical wireless powered MEC systems under stochastic fading channels, we need to make online decisions to optimize the long-term system performance under future channel uncertainty. The corresponding optimal system design is very challenging. First, the online decisions of WPT and task processing may not meet the long-term performance requirements such as the average power constraint of the HAP. Second, the control decisions are tightly coupled over time due to the temporal correlations of battery and data buffer dynamics. It can lead to extremely low energy efficiency if we independently optimize the system performance within each time slot in a greedy manner. Moreover, the fast-varying channel requires a low-complexity solution method to facilitate real-time implementation, which is often difficult for a complicated optimal control problem. }


\begin{figure}[h]
	\vspace{-15pt}
	\centering
	\includegraphics[scale=0.66]{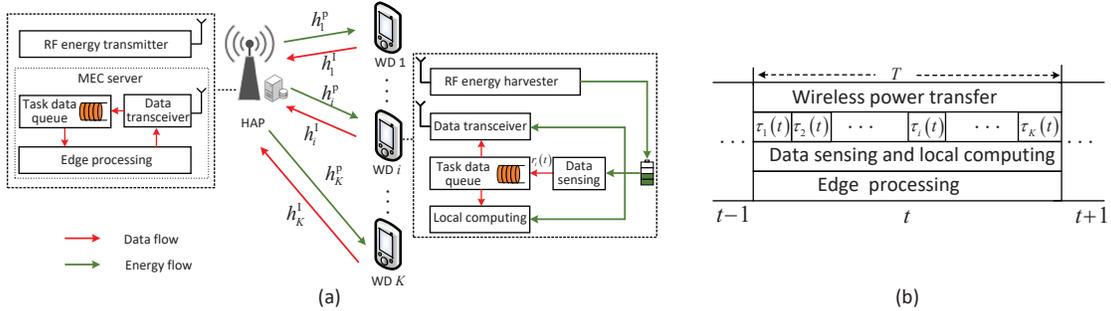}
	\captionsetup{font=footnotesize}
	\caption{The considered wireless powered MEC system: (a) the system model and (b) a time allocation example.}
	\label{Sys_mod}
	\vspace{-12pt}
\end{figure}

In this paper, we focus on designing an energy-efficient multi-user wireless powered MEC system under fading channels. In particular, we consider in Fig.\ref{Sys_mod} that multiple WDs harvesting RF energy broadcast from an HAP, and using the harvested energy to sense and process the task data assisted by the edge server. Under an average power constraint of the HAP, we aim to maximize the long-term average data sensing rate of WDs, while maintaining data queue stability at all the WDs and the HAP. The main contributions of this paper are:

1) \emph{Energy-efficient WPT and Task Processing Optimization:} 
	{\color{black}Under stochastic fading channels, we jointly optimize the WPT and task processing strategies of the wireless powered edge computing system subject to an average power constraint of the HAP. We formulate the target problem as a multi-stage stochastic optimization. The major difficulty lies in seeking a real-time online solution to satisfy all the long-term constraints under the randomness of fading channels and strong couplings among solutions in different time slots.}
	
2) \emph{Low-complexity Online Algorithm:} We propose an online algorithm named LEESE without any knowledge of future channel fading information. In particular, LEESE leverages the perturbed Lyapunov optimization to equivalently transform the multi-stage stochastic optimization into per-slot deterministic optimization problems, each decides the system control variables in the current time slot. We show that each per-slot problem can be equivalently transformed into a convex optimization problem. To reduce the computation overhead of online implementation, instead of solving each per-slot problem with off-the-shelf convex optimization (CVX) algorithms, we design an efficient block coordinate descent (BCD)-based method and derive closed-form solutions of simple threshold-based structure. 

3) \emph{Theoretical Performance Analysis:} We show that under a moderate battery capacity of each WD, LEESE always produces a feasible, and in fact asymptotically optimal, solution to the original multi-stage stochastic problem. Meanwhile, we prove that LEESE achieves an $\left[O(1/V),O(V)\right]$ sensing rate-delay tradeoff by tuning a non-negative control parameter $V$. Specifically, setting a larger $V$ yields higher sensing rate at cost of longer data processing delay, and vice versa. 

4) \emph{Numerical Performance Evaluation:} We verify the performance of LEESE via extensive numerical simulations. The results show that LEESE can provide more than 21.9\% higher sensing rate over the considered benchmark methods. Besides, we find that the BCD-based LEESE algorithm achieves almost identical data sensing performance to the CVX-based LEESE, while incurring only 0.04\% of the computation delay. 

\vspace{-15pt}
\subsection{Related works}
There have been extensive research interests in optimizing the long-term performance of MEC systems with energy-harvesting WDs, where the WDs harvest energy either from ambient renewable such as solar and wind \cite{Mao2016,Min2019,XIANarxiv2021,Xu2017a,Wei2019}, or dedicated RF energy transmitters \cite{Wang2020c,Wu2019,Sun2021,Deng2020}. 

For the former line of works, \cite{Mao2016} considered an MEC network with an MEC server assisting a single EH-WD, and minimized the long-term execution delay and task dropping cost by jointly optimizing the task offloading decision, data transmission power and local CPU frequency at the WD. \cite{Min2019} studied the optimal server selection and task allocation in a multi-server MEC network with single EH WD. \cite{XIANarxiv2021} extended the research into spectrum scarcity scenario and studied the optimal cognitive online sensing and processing of the EH WD. Considering an EH-powered base station co-located with multiple edge servers, \cite{Xu2017a} studied the optimal task offloading and server autoscaling policy to minimize the energy cost and processing delay. \cite{Wei2019} formulated an offloading control problem of a single EH-WD as a Markov decision process (MDP), and designed a reinforcement learning (RL) based method to optimize the long-term system utility concerning computation rate and processing latency.\

To tackle the random ambient energy arrivals, the second line of research applies WPT technology to achieve controllable power supply to WDs for sustainable long-term computation performance. For example, \cite{Wang2020c} focused on the offline optimal energy allocation and task offloading of a single-WD wireless powered MEC system. \cite{Wu2019} introduced a stochastic modeling of the edge computing process, and designed a Lyapunov optimization-based online algorithm to maximize the long-term system computation rate. Aiming at optimizing system energy efficiency, \cite{Sun2021} studied the dynamic control on WPT and resource allocation among device-to-device (D2D)-assisted WDs. Nonetheless, these works either assume non-causal prior knowledge of the system state \cite{Wang2020c}, or {\color{black}sufficient energy supply} at the HAP \cite{Sun2021,Wu2019}. When the HAPs rely on harvesting ambient energy to perform edge computation, \cite{Deng2020} considered multiple EH-assisted HAPs and optimized the user association and resource allocation under stochastic energy and data arrivals. However, \cite{Deng2020} considered a simplified system setup where each HAP can accept tasks from only a single WD. In our considered setup with causal system information and average HAP power constraint, {\color{black}the optimal design is challenged by the unknown future channel information and the real-time control of tightly-coupled decisions on WPT and task execution over different time slots.} This calls for delicate coordination of the HAP and WDs where the solutions in existing works \cite{Wang2020c,Wu2019,Sun2021,Deng2020} are no longer applicable.

The remainder of this paper is organized as follows. In Section II, we present the system model of the wireless powered MEC system and formulate the sensing rate maximization problem. In Section III, we propose the LEESE method to solve the problem. We prove the feasibility and asymptotic optimality of LEESE in Section IV and evaluate the performance of LEESE via numerical simulations in Section V. Finally, we conclude the paper in Section VI.

\vspace{-10pt}
\section{System Model and Problem Formulation}\label{sec3}
As shown in Fig. \ref{Sys_mod}(a), we consider a wireless powered MEC networks consisting of $K$ WDs and one HAP. The HAP is connected to a stable power grid, while each WD is equipped with an energy harvesting module and a rechargeable battery. Integrated with an RF energy transmitter and an MEC server, the HAP provides energy supply and task computation assistance for WDs in sequential time slots of equal duration $T$. In particular, in time slot $t$, the HAP broadcasts RF energy to the WDs, while each WD harvests energy and stores it into the battery. Meanwhile, with the stored energy, each WD takes raw data samples from the monitored environment and piles it into the local task data queue. To process the raw task data, each WD employs a partial computation offloading rule which allows the raw task data to be arbitrarily partitioned into two parts with one executed locally and the other offloaded to the HAP for edge computing \cite{Wang2016}. 

We consider block fading channels for both WPT and data communication, where channel gains between the HAP and WDs are assumed static within each time slot, but may vary from one slot to another. We assume that WPT and data communication are implemented over orthogonal frequency bands, thus they can be performed simultaneously. To avoid co-channel communication interference among the WDs, we assume that WDs communicate with the HAP using a TDMA scheme. In Fig. \ref{Sys_mod}(b), we show an example time allocation of the considered system in time slot $t$. Specifically, the HAP performs WPT and edge processing throughout the duration $T$, while WD$_i$ performs task offloading within the allocated time frame $\tau_i(t)$, $i=1,\cdots,K$. Since circuits of data sensing, communication and local processing are separated with each other, WD$_i$ can simultaneously perform data collection , task offloading and local computing. In addition, we neglect the time cost on result feedback from the HAP to WDs assuming that the computation result is relatively short (e.g., several bits feedback for identifying the objects in a picture).

\vspace{-15pt}
\subsection{Wireless Power Transfer}
In the $t$-th time slot, the HAP broadcasts wireless energy to WDs with transmit power $p_0(t)\leq p_0^{\rm max}$, where $p_0^{\rm max}$ is the maximum WPT power at the HAP. The energy cost on WPT is
\begin{equation}\label{WPT_power}
\small
e_{0,\rm P}(t) = p_0(t)T.
\end{equation}
Denote $h_i^{\rm p}(t)$ as the WPT channel gain from the HAP to the $i$th WD. We apply a practical non-linear energy harvesting model \cite{Chen2017} to depict the energy conversion efficiency at WDs. Specifically, the energy harvested by the $i$th WD is
\begin{equation}
\small
e_{i,\rm h}(t) = \frac{a_{1,i}p_0(t)h_i^{\rm p}(t)+a_{2,i}}{p_0(t)h_i^{\rm p}(t)+a_{3,i}}-\frac{a_{2,i}}{a_{3,i}},
\end{equation}
where $a_{1,i}$, $a_{2,i}$, $a_{3,i}$ are constant parameters for $i=1,2,\cdots,K$. {\color{black}$e_{i,\rm h}(t)$ has the following property.

\begin{Lemma}\label{lem_convex}
	The energy harvesting parameters $a_{1,i}$, $a_{2,i}$, and $a_{3,i}$ satisfy that $a_{1,i}a_{3,i}\!-\!a_{2,i}\!\geq \!0$ and $a_{3,i}\!>\!0$. The harvested energy $e_{i,\rm h}(t)$ is a concave function of the wireless charging power $p_0$.
\end{Lemma}
\begin{proof} Please refer to Appendix \ref{app_convex} for detail.
\end{proof}}

\vspace{-15pt}
\subsection{Task Data Sensing and Processing at WDs}
\emph{1) Task Data Sensing}: In time slot $t$, the WD$_i$ collects $r_i(t)$ bits of data samples from the monitored environment. $r_i(t)\leq r_{\rm max}$ holds due to the constrained sampling resolution/frequency. Denote $e_{\rm unit}^{\rm col}$ in Jolue/bit as the unit energy cost for data sensing. Then, the energy consumption on collecting $r_i(t)$ bits of data is\cite{Liu2020}
\begin{equation}
e_{i,\rm col}(t)=e_{\rm unit}^{\rm col}r_i(t).
\end{equation}
The WD$_i$ stores the sensed data in a local buffer, and subsequently processes each task data bit via either local computing or computation offloading.

\emph{2) Local Computing}: We denote the CPU frequency of WD$_i$ in time slot $t$ as $f_i(t)$ in cycles/second, where $f_i(t)\leq f_{\rm i}^{\rm max}$ due to the local computation capability constraint. Then, the task data processed locally and the corresponding energy consumption are \cite{Wang2018}
\begin{equation}\label{eq_time_cost_comp_loc} 
\small
D_{i,\rm C}(t) = f_i(t)T/\phi_i, ~e_{i,\rm C}(t) = \kappa_{i}\left(f_i(t)\right)^3T,
\end{equation}
respectively, where $\phi_i$ denoted the required local CPU cycles to one bit of data and $\kappa_i$ is the energy efficiency parameter for local computing. 

\emph{3) Computation Offloading}: The WDs offload their computation tasks to the
HAP in a TDMA manner. Denote $h_i^{\rm I}(t)$ as the computation offloading channel gain from the WD$_i$ to the HAP, $p_i(t)$ as the transmit power constrained by its maximum value $p_i(t)\leq p_i^{\rm max}$ and $\tau_i(t)$ as the allocated time for task offloading. The task data offloaded by the WD$_i$ is
\begin{equation}
\small
D_{i,\rm O}(t) = {W\tau_i(t)}\log_2\left(1+{p_i(t)\gamma_i(t)}\right),
\end{equation}
where $W$ is the task offloading bandwidth. $\gamma_i(t)=\frac{h_i^{\rm I}(t)}{N_0}$ and $N_0$ denotes additive white Guassian noise (AWGN) power at the HAP.
Correspondingly, the energy cost on task offloading is 
\begin{equation}
e_{i,\rm O}(t)=p_i(t)\tau_i(t),
\end{equation}
By summing up the energy cost on data sensing, local processing and computation offloading, the total energy consumption of WD$_i$ in time slot $t$ is
\begin{equation}
e_i(t) = e_{i,\rm col}(t) + e_{i,\rm C}(t) + e_{i,\rm O}(t).
\end{equation}

\vspace{-15pt}
\subsection{Task Computation at the HAP}
We denote $f_0(t)$ as the edge CPU frequency at the HAP in time slot $t$, where $f_0(t)$ is upper bounded by $f_0(t)\leq f_0^{\rm max}$. The task data processed at edge and the corresponding energy consumption on edge computing are
\begin{equation}\label{eq_time_cost_comp_edg} 
\small
D_{0,\rm C}(t) = f_0(t)T/\phi_0, ~e_{0,\rm C}(t) = \kappa_0\left(f_0(t)\right)^3T,
\end{equation}
respectively, where $\phi_0$ denotes the required edge CPU cycles to process one bit of data and $\kappa_0$ is the energy efficiency parameter for edge computing. By summing up the energy cost on WPT and edge computing, the total energy cost of the HAP in time slot $t$ is
\begin{equation}
e_0(t) = e_{0,\rm P}(t) + e_{0,\rm C}(t).
\end{equation}

\vspace{-15pt}
\subsection{Task Data Queue Model}
For the HAP and WDs, the data received or sensed in one time slot is ready for processing at the beginning of the next time slot. Denote $Q_i(t)$ and $Q_0(t)$ as the data queue length of the WD$_i$ and HAP at the start of time slot $t$, respectively. The data processed at the WD$_i$ and HAP within the current time slot must satisfy the data causality constraints:
\begin{equation}\label{ledge_off_cons}
0\leq D_{0,\rm C}(t)\leq Q_0(t),~~ 0\leq D_{i,\rm C}(t)+D_{i,\rm O}(t) \leq Q_i(t), \forall i, t.
\end{equation}
We assume infinite task queue capacity and focus on the asymptotic stability of data queues. Then, the dynamics of $Q_i(t)$ and $Q_0(t)$ over time are
\begin{equation}\label{Data_queue_simp}
\small
Q_i(t+1) = Q_i(t)-D_{i,\rm C}(t)-D_{i,\rm O}(t) + r_i(t),~Q_0(t+1) = Q_0(t) - D_{0,\rm C}(t) + \mathsmaller{\sum_{i=1}^{K}}D_{i,\rm O}(t), \forall t,
\end{equation}
respectively. To ensure stable data queues at the HAP and all WDs, we consider following stability constraints on $Q_0(t)$ and $Q_i(t)$ \cite{Neely2010}: 
\begin{equation}\label{DataQ_stab}
\small
\bar{Q}_0 = \lim_{N\to+\infty}\frac{1}{N}\mathsmaller{\sum_{t=1}^{N}}\mathbb{E}\left[Q_0(t)\right]<\infty,~\bar{Q}_i = \lim_{N\to+\infty}\frac{1}{N}\mathsmaller{\sum_{t=1}^{N}}\mathbb{E}\left[Q_i(t)\right]<\infty, \forall i=1, \cdots, K, 
\end{equation}
where the expectation is taken over the time-varying channels.
 
\vspace{-15pt}
\subsection{Energy Queue Model}	
Denote $B_i(t)$ as the battery level of the WD$_i$ at the start of time slot $t$. The WD$_i$ adopts an energy-aware battery management policy to prevent permanent device failure due to full battery depletion: when $B_i(t)$ is below a threshold $B_{\rm min}$, it stops consuming energy on data sensing or processing, while only receiving RF energy from the HAP to replenish the battery. Intuitively, the energy consumed by the WD$_i$ within the current time slot must satisfy
\begin{equation}\label{EH_caus}
0\leq \lambda_{\rm e}e_i(t)\leq B_t\cdot\mathbbm{1}_{B_i(t)\geq B_{\rm min}}, \forall i, t,
\end{equation}
where $\lambda_{\rm e}\geq 1$ is a scaling factor (e.g., $\lambda_e = 1000$ when we use mJ as the unit) and $\mathbbm{1}_{\{\cdot\}}$ denotes the indicator function. The battery queue of the WD$_i$ evolves as:
\begin{equation}\label{E_evol}
B_i(t+1) = \min\left(B_i(t) - \lambda_{\rm e}e_i(t) + \lambda_{\rm e}e_{i,\rm h}(t), \Omega_i\right),
\end{equation}
where $\Omega_i$ is the battery capacity of WD$_i$. 

\vspace{-15pt}
\subsection{Problem Formulation}
Denote the system state at the beginning of time slot $t$ as $s(t)=\{I(t), h_i^{\rm P}(t), h_i^{\rm I}(t),\forall i\}$ at the HAP, where $I(t)=\{Q_0(t),Q_i(t),B_i(t)\}$ denotes the queue backlog state. Our objective is to maximize the long-term average data sensing rate of all WDs while satisfying system stability and average power constraints of the HAP. Under time-varying system states, the wireless powered MEC system requires judicious control of WPT, data sensing and task computation in sequential time slots. We formulate the optimization problem as below:
\begin{subequations}\label{Primal_prob}
	\small
	\begin{align}
	\text{(P1)}~~\underset{\substack{r_i(t), p_i(t),f_i(t), \\ p_0(t), f_0(t),\forall i, t}} \max&~\bar{R}=\lim_{N\to+\infty}\frac{1}{N}\mathsmaller{\sum_{t=0}^{N-1}\sum_{i=1}^{K}}w_ir_i(t)\\
	\st
	~~& \eqref{ledge_off_cons}, \eqref{DataQ_stab}, \eqref{EH_caus},\\
	~~& \mathsmaller\sum_{i=1}^{K} \tau_i(t) \leq T, \forall t, \label{T_allo_cons}\\
	~~& \lim_{N\to+\infty}\frac{1}{N}\mathsmaller\sum_{t=1}^{N}\mathbb{E}\left[e_0(t)\right] \leq e_0^{\rm th}, \label{Bud_cons}\\
	~~& 0 \!\leq r_i(t) \!\leq\! r_{\rm max}, 0 \!\leq p_i(t) \!\leq\! p_i^{\rm max},\! 0 \leq\! f_i(t) \leq\! f_{i}^{\rm max}, \!0 \leq\! f_0(t)\! \leq\! f_0^{\rm max}, \forall i, t, \label{Prob_const2}
	\end{align}
\end{subequations}
where $w_i$ is the weighting factor describing the importance of WD$_i$. \eqref{ledge_off_cons} depicts the data causality at the HAP and WDs. \eqref{EH_caus} is the energy causality at WDs. \eqref{DataQ_stab} captures the data queue stability. \eqref{T_allo_cons} constraints the time allocation on task offloading. \eqref{Bud_cons} represents that the average power constraint at the HAP should be lower than the threshold $e_0^{\rm th}$. We aim at designing an online algorithm which determines the actions in time slot $t$ (i.e, $\{r_i(t), p_i(t),f_i(t), p_0(t), f_0(t), \forall i\}$) based only on $s(t)$, i.e., without the knowledge of future information. However, such an online design faces two major challenges. On one hand, under the stochastic channels, it is difficult to satisfy the long-term requirements for the decisions made in sequential time slots without future channel information. On the other hand, the control decisions of the HAP and WDs are inherently coupled in terms of energy consumption. This poses a great challenge to jointly optimize the decisions over a long time span to achieve a good balance between the current and future system performance. In the following, we propose a perturbed-\textbf{L}yapunov-based \textbf{E}nergy-\textbf{E}fficient data \textbf{S}ensing and \textbf{E}dge computation (LEESE) algorithm to solve (P1), which achieves an online control on wireless power transfer, data sensing and processing without requiring a prior system knowledge.

\vspace{-1em}
\section{Online Data Sensing and Computation Offloading Optimization}\label{sec4}
\subsection{Perturbed Lyapunov-based Optimization}
Lyapunov optimization is extensively applied in stochastic MEC systems to ensure long-term stability \cite{Mao2016,Deng2020,Sun2021,Wu2019}. Under stable power supply, the vanilla Lyapunov optimization technique in these works relies crucially on the assumption that the feasible control action set is irrelevant to the energy state. However, this condition is violated in the considered wireless powered MEC system where the available control actions critically depend on the battery level under the temporally correlated energy constraint \eqref{EH_caus}. As a result, the vanilla Lyapunov optimization cannot be directly used to solve (P1). In the following, we introduce a perturbed Lyapunov method to tackle this problem. First, we introduce a perturbed battery queue for each WD, i.e., 
\begin{equation}\label{Queue_B}
\tilde{B}_i(t) \triangleq  B_i(t) - \Omega_i, \forall i.
\end{equation} 
The purpose of perturbation is to push the target battery level at the WD$_i$ toward $\Omega_i$. Then, as shown in Section \ref{secV}, the battery energy constraint \eqref{EH_caus} becomes implicit when we employ a sufficiently large $\Omega_i$. To ensure the average power requirement at the HAP, we define a virtual energy deficit queue with dynamics
\begin{equation}\label{Queue_Z}
Z_0(t+1) =  \max\left(Z_0(t) + \lambda_{\rm c} e_0(t) - \lambda_{\rm c} e_0^{\rm th},0\right), \forall t,
\end{equation}
where $\lambda_{\rm c}$ is a positive scaling factor. Then, we can satisfy the average power constraint \eqref{Bud_cons} by stabilizing $Z_0(t)$ \cite{Neely2010}. Let $\Theta(t)\triangleq\left\{\tilde{B}_i(t),Z_0(t),Q_i(t),Q_0(t), \forall i\right\}$ be the system queue backlog. We define the perturbed Lyapunov function as
\begin{equation}
\Phi(t) = \frac{1}{2}\mathsmaller{\sum_{i=1}^{K}}\left[\tilde{B}_i(t)\right]^2 + \frac{1}{2}\mathsmaller{\sum_{i=1}^{K}}\left[Q_i(t)\right]^2 + \frac{1}{2}\left[Z_0(t)\right]^2 + \frac{1}{2}\left[Q_0(t)\right]^2,
\end{equation}
and the Lyapunov drift-plus-penalty function as
\begin{equation}\label{LyaDrift}
\Delta^t_V = \mathbb{E}\left[\Phi(t+1)-\Phi(t)|\Theta_t\right]-V\mathbb{E}\left[\mathsmaller{\sum_{i=1}^{K}}r_i(t)|\Theta(t)\right],
\end{equation}
where $V$ is a positive weight factor. The expectation is taken over the random channel state given the current queue backlog $\Theta(t)$. For convenience, we denote a constant
\begin{equation}\label{Dvalue}
\small
\begin{split}
C=&\frac{1}{2}\left\{\mathsmaller{\sum_{i=1}^{K}}\left[\left(D_{i,\rm C}^{\rm max}+D_{i,\rm O}^{\rm max}\right)^2+\left(r_{\rm max}\right)^2\right] + \left(\mathsmaller{\sum_{i=1}^{K}}D_{i,\rm O}^{\rm max}\right)^2 + \left(D_{0,\rm C}^{\rm max}\right)^2 \right. + \\
&\left.\mathsmaller{\sum_{i=1}^{K}}\left[(\lambda_{\rm e}e_i^{\rm max})^2+\left(\lambda_{\rm e}e_{i,\rm h}^{\rm max}\right)^2\right] + \left(\lambda_{\rm c}e_0^{\rm max}\right)^2+\left(\lambda_{\rm c}e_0^{\rm th}\right)^2\right\},
\end{split}
\end{equation}
where $D_{i,\rm C}^{\rm max}\triangleq\frac{f_i^{\rm max}T}{\phi_i}$ and $D_{0,\rm C}^{\rm max}\triangleq\frac{f_0^{\rm max}T}{\phi_0}$ are the maximum per-slot processing data via local computing and edge execution, respectively. $D_{i,\rm O}^{\rm max}\!\triangleq\! \mathbb{E}\left[WT\log_2\left(1\!+\!p_i^{\rm max}\gamma_t\right)\right]$ is the maximum average offloading rate of the WD$_i$. 
$e_i^{\rm max}\!\triangleq\!e_{\rm unit}^{\rm col}r_{\rm max}\!+\!p_i^{\rm max}T\!+\!\kappa_i\left(f_i^{\rm max}\right)^3T$ denotes the maximum per-slot energy drain at the WD$_i$. $e_0^{\rm max}\triangleq\kappa_0\left(f_0^{\rm max}\right)^3T+p_0^{\rm max}T$ is the maximum per-slot energy cost at the HAP. Then, we have the following lemma regarding $\Delta^t_V$. 
\vspace{-5pt}
\renewcommand{\algorithmicrequire}{\textbf{Initialization:}}
\renewcommand{\algorithmicensure}{\textbf{Output:}}
\begin{algorithm}[H]
	\small
	\caption{The online LEESE algorithm to solve (P1)} 
	\label{alg1}
	\begin{algorithmic}[1]
		\REQUIRE The initial system state $s(t)=\{\Theta(0), h_i^{\rm P}(0), h_i^{\rm I}(0),\forall i\}$, where $\Theta(0) = \{Q_{0}(0),Z_0(0),Q_i(0),\tilde{B}_i(0),\forall i\}$.
		\FOR {each time slot $t$}
		\STATE Observe the system state $s(t)$.
		\STATE Solve problem \eqref{Prob_Per_slot} for $\{p_0^\ast(t), f_0^{\ast}(t), r_i^\ast(t), f_i^{\ast}(t), p_i(t)^{\ast}, \tau_i^\ast(t),\forall i\}$. Specifically, obtain $p_0^\ast(t)$, $f_0^{\ast}(t)$, $r_i^\ast(t)$, $\forall i$, using \eqref{Opt_p0}, \eqref{Opt_fs} and \eqref{Opt_r}, and $f_i^{\ast}(t), p_i(t)^{\ast}, \tau_i^\ast(t)$, $\forall i$ by solving \eqref{Prob_Per_slot_sim}, respectively.
		\STATE Execute the control action $\{p_0^\ast(t), f_0^{\ast}(t), r_i^\ast(t), f_i^{\ast}(t), p_i(t)^{\ast}, \tau_i^\ast(t)\}$ and update $\Theta(t+1)=\{Q_{0}(t+1),Z_0(t+1),Q_i(t+1),\tilde{B}_i(t+1),\forall i\}$ according to \eqref{Data_queue_simp}, \eqref{Queue_B} and \eqref{Queue_Z}, respectively.
		\ENDFOR
	\end{algorithmic}
\end{algorithm}
\vspace{-20pt}

\begin{Lemma}\label{lem6}
	Given any feasible control actions and any queue backlogs, the Lyapunov drift-plus-penalty function $\Delta^t_V$ is upper bounded by
	\begin{equation} \label{LyaDriPlusPen}
	\begin{split}
	\Delta^t_V &\leq C-\! \mathbb{E}\left\{L(t)\mid\!\!\Theta_t\right\},
	\end{split}
	\end{equation}
	where $L(t) = V\mathsmaller{\sum_{i=1}^{K}}w_ir_i(t) \!+ \!\mathsmaller{\sum_{i=1}^{K}}Q_i(t)\left(D_{i,\rm C}(t)\!+\!D_{i,\rm O}(t)\!-\!r_i(t)\right) + Q_0(t)\left(D_{0,\rm C}(t)-\mathsmaller{\sum_{i=1}^{k}}D_{i,\rm O}(t)\right)+ \mathsmaller{\sum_{i=1}^{K}}\lambda_{\rm e}(B_i(t)-\Omega_i)(e_i(t)-e_{i,\rm h}(t)) + Z_0(t)\lambda_{\rm c}\left(e_0^{\rm th}-e_0(t)\right)$.
\end{Lemma}
\begin{proof}
	Please refer to Appendix \ref{app0} for detail.
\end{proof}

The key idea of our proposed LEESE algorithm is to greedily minimize the right-hand-side (RHS) of \eqref{LyaDriPlusPen} in each time slot. The intuition behind this operation is that by minimizing $\Delta^t_V$, we not only push all the queues in $\Theta(t)$ towards zero but also maximize the data sensing rate $\bar{R}$. We illustrate the procedures of LEESE in Algorithm \ref{alg1}. In time slot $t$, LEESE observes the current system states $s(t)$ and determines the actions of the HAP and WDs by solving per-slot optimization problem
\begin{subequations}\label{Prob_Per_slot}
\small
\begin{align}
\underset{\substack{p_0, f_0, r_i, p_i,f_i, \forall i,}} \max~&V\mathsmaller{\sum_{i=1}^{K}}w_ir_i \!+ \!\mathsmaller{\sum_{i=1}^{K}}Q_i\left(D_{i,\rm C}\!+\!D_{i,\rm O}\!-\!r_i\right) + Q_0\left(D_{0,\rm C}-\mathsmaller{\sum_{i=1}^{k}}D_{i,\rm O}\right) \nonumber\\
&+ \mathsmaller{\sum_{i=1}^{K}}\lambda_{\rm e}(B_i-\Omega_i)(e_i-e_{i,\rm h}) + Z_0\lambda_{\rm c}\left(e_0^{\rm th}-e_0\right)\label{Prob_Per_slot_obj}\\
\st
~~& D_{i,\rm C}\!+\!D_{i,\rm O} \leq Q_i, \forall i,\\
~~& \mathsmaller\sum_{i=1}^{K} \tau_i \leq T,\\
~~& 0\leq f_0T/\phi_0\leq Q_0, \label{Prob_Per_slot_cons3}\\
~~& 0 \!\leq \!r_i \!\leq \!r_{\rm max}, 0 \!\leq \!p_i \!\leq \!p_i^{\rm max}, \!0 \leq\! f_i \!\leq\! f_i^{\rm max}, 0 \!\leq \!f_0\! \leq \!f_0^{\rm max}, 0 \!\leq \!p_0\! \leq \!p_0^{\rm max}, \forall i,
\end{align}
\end{subequations}
where we drop the time index ``$(t)$'' for brevity. Comparing with \eqref{Primal_prob}, we remove the energy causality constraint \eqref{EH_caus} in the per-slot problem \eqref{Prob_Per_slot}. In Section \ref{secV}, we show that LEESE can always respect \eqref{EH_caus} given that the battery capacity of each WD satisfies a mild condition. In the following, we obtain the optimal solution to \eqref{Prob_Per_slot} by solving four independent subproblems, which correspond to wireless charging power control at the HAP, CPU frequency allocation at the HAP, data sensing control at WDs, and task execution control at WDs, respectively.

\vspace{-15pt}
\subsection{Optimal Wireless Charging Power}
A close observation of \eqref{Prob_Per_slot} shows that wireless charging power at the HAP can be separately optimized by solving the following problem
\begin{subequations}\label{Subprob_p0}
	\begin{align}
	\underset{\substack{p_0}} \max~& G_1(p_0) = -\mathsmaller{\sum_{i=1}^{K}}\lambda_{\rm e}\left(B_i-\Omega_i\right)e_{i,\rm h} - Z_0\lambda_{\rm c}p_0T \label{Subprob_p0_obj}\\
	\st
	~~& 0 \leq p_0 \leq p_0^{\rm max}. \label{Subprob_p0_constr}
	\end{align}
\end{subequations}
Because $\Omega_i\geq B_i$, Lemma \ref{lem_convex} shows that \eqref{Subprob_p0_obj} is a concave function in $p_0$. When $Z_0\!=\!0$, $G_1(p_0)$ monotonically increases with $p_0$ and thus $p_0\!\!=\!p_0^{\rm max}$. When $Z_0\!\!>\!\!0$, the first derivative of $G_1\!(p_0)$ is  
\begin{equation}\label{G_solu}
G_1^\prime(p_0)=\frac{\partial G_1(p_0)}{\partial p_0} = \mathsmaller{\sum_{i=1}^{K}}\lambda_{\rm e}\left(\Omega_i-B_i\right)\frac{\left(a_{1,i}a_{3,i}-a_{2,i}\right)h_i}{\left(p_0h_i+a_{3,i}\right)^2}-Z_0\lambda_{\rm c}T=0.
\end{equation}
Since $a_{1,i}a_{3,i}\geq a_{2,i}$, $G^\prime(p_0)$ monotonically decreases with the increasing of $p_0\in[0, \infty)$. As $p_0\to \infty$, $G^\prime(p_0)=-Z_0\lambda_{\rm c}T\leq 0$. In the following, we discuss the solution of \eqref{Subprob_p0} in two cases: 1) When $G_1^\prime(0)>0$, \eqref{G_solu} has a unique solution $\hat{p}_0\in[0, \infty)$, where $\hat{p}_0$ can be efficiently obtained via bi-section search method. In this case, the optimal wireless charging power $p_0^\ast = \min(\hat{p}_0,p_0^{\rm max})$. 2) When $G_1^\prime(0)\leq 0$, $G_1^\prime(p_0)\leq Z_0\lambda_{\rm c}T$ for all $p_0\geq 0$. In this case, \eqref{Subprob_p0_obj} is a monotonically decreasing function of $p_0$ and thus $p_0^\ast=0$. In summarize, the optimal solution to \eqref{Subprob_p0} is
\vspace{2pt}
\begin{small}
\begin{subnumcases}{p_0^{\ast} =\label{Opt_p0}}
\min(\hat{p}_0,p_0^{\rm max}), \!& \text{if $G_1^\prime(0)>0$ and $Z_0>0$}, \label{Opt_p0_1}\\
0, \!& \text{if $G_1^\prime(0)\leq 0$ and $Z_0>0$}, \label{Opt_p0_2}\\
p_0^{\rm max}, \!& \text{if $Z_0=0$}. \label{Opt_p0_3}
\end{subnumcases}
\end{small}\vspace{2pt}\eqref{Opt_p0} reveals that the optimal wireless charging follows a threshold structure: the HAP broadcasts power to WDs when $G_1^\prime(0)>0$, and shuts down the wireless power charging circuit when $G_1^\prime(0)\leq 0$. The threshold $G_1^\prime(0)$ decreases with the growth of battery level $B_i$ and power deficit $Z_0$. In the special case of $Z_0=0$ (i.e., the power budget of HAP is sufficient), the HAP broadcasts energy to WDs at the maximum transmit power $p_0^{\rm max}$. 

\subsection{Optimal Edge CPU Frequency}
Similarly, we can independently optimize the CPU frequency of the HAP by solving the following convex optimization:
\begin{subequations}\label{Subprob_fs}
	\begin{align}
	\underset{\substack{f_0}} \max~& - Z_0\lambda_{\rm c} \kappa_0\left(f_0\right)^3T + Q_0f_0T/\phi_0 \label{Subprob_fs_obj}\\
	\st
	~~&0\leq f_0T/\phi_0\leq Q_0,~0 \leq f_0 \leq f_0^{\rm max},\label{Subprob_fs_const}
	\end{align}
\end{subequations}
where the optimal solution is
\begin{equation}\label{Opt_fs}
f_0^{\ast}=\min\left(\sqrt{{Q_0}/{\left(3Z_0\lambda_{\rm c} \phi_0\kappa_0\right)}},\bar{f}_0^{\rm max}\right).
\end{equation}
Here, $\bar{f}_0^{\rm max}=\min(Q_0\phi_0/T, f_0^{\rm max})$. As shown in \eqref{Opt_fs}, the edge CPU frequency increases with data queue length $Q_0$ and decreases with the power deficit queue $Z_0$. Together with the optimal charging power control, such an edge CPU frequency scheduling tends to stabilize the data queue $Q_0$ and satisfy the average power constraint at the HAP.

\subsection{Optimal Data Sensing Rate}
We obtain the optimal task data size collected in time slot $t$ by solving the following linear programming (LP) problem:
\begin{equation}\label{Opt_data}
\underset{\substack{0\leq r_i\leq r_{\rm max}, \forall i}}\max \mathsmaller{\sum_{i=1}^{K}}\left(Vw_i+\lambda_{\rm e}(B_i\!-\!\Omega_i)e_{\rm unit}^{\rm col}-Q_i\right)r_i.
\end{equation}
The optimal solution of \eqref{Opt_data} exhibits a simple ON-OFF structure:
\begin{equation}
r^{\ast}_i = r_{\rm max}\cdot \mathbbm{1}_{C_i^{\rm sen}\leq 0}, \forall i, \label{Opt_r}
\end{equation}
where $C_i^{\rm sen}\triangleq Q_i-Vw_i-\lambda_{\rm e}(B_i\!-\!\Omega_i)e_{\rm unit}^{\rm col}$. In particular, the WD$_i$ operates at the maximum data sensing rate (i.e., $r_i=r_{\rm max}$) when $C_i^{\rm sen}\leq 0$, and stops data sampling otherwise. Because $C_i^{\rm sen}$ grows with $Q_i$ and falls with $B_i$, the WD$_i$ reduces sensing activity when $Q_i$ is large or $B_i$ is small, thus avoiding continuous local data backlog and energy draining.

\subsection{Optimal Task Execution}
The task execution at WDs includes local computing and task offloading. By removing the terms that are only related to $p_0$, $f_0$ and $r_i$ in \eqref{Prob_Per_slot}, we solve the following optimization problem to optimize task execution of all the WDs
\begin{subequations}\label{Prob_Per_slot_sim}
	\small
	\begin{align}
	\underset{\substack{p_i, f_i, \tau_i,\forall i}} \max~&\mathsmaller{\sum_{i=1}^{K}}D_{i,\rm O}\left(Q_i-Q_0\right) + \mathsmaller{\sum_{i=1}^{K}}Q_iD_{i,\rm C} + \mathsmaller{\sum_{i=1}^{K}}\lambda_{\rm e}\tilde{B}_i\left(e_{i,\rm O}+e_{i,\rm C}\right) \label{Prob_Per_slot_sim_redc_obj}\\
	\st
	~~&D_{i,\rm C} + D_{i,\rm O} \leq Q_i, \forall i \label{Prob_Per_slot_sim_cons_data}\\
	~~&\mathsmaller{\sum_{i=1}^{K}}\tau_i \leq T, \label{Prob_Per_slot_sim_cons_time}\\
	~~& 0 \!\leq \!p_i \!\leq \!p_i^{\rm max}, ~0 \leq\! f_i \!\leq\! f_i^{\rm max}, ~0\leq \tau_i \leq T, \forall i. \label{Prob_Per_slot_sim_cons_pft}
	\end{align}
\end{subequations}
Generally, \eqref{Prob_Per_slot_sim} is a non-convex problem due to the time-varying coefficient $\left(Q_i-Q_0\right)$ and non-convex constraint \eqref{Prob_Per_slot_sim_cons_data}. Nonetheless, we show that \eqref{Prob_Per_slot_sim} can be equivalently transformed into a convex optimization as follows. For convenience, we denote the set of WDs as $\mathcal{W}$. We introduce auxiliary variables $r_{i,\rm O}$'s and rewrite \eqref{Prob_Per_slot_sim} as
\begin{subequations}\label{Prob_Per_slot_sim_rew}
	\small
	\begin{align}
	\underset{\substack{f_i, \tau_i, r_{i,\rm O}, e_{i,\rm O},\forall i}} \max~&\mathsmaller{\sum_{i\in\mathcal{W}}}r_{i,\rm O}\left(Q_i-Q_0\right) + \mathsmaller{\sum_{i\in\mathcal{W}}}Q_iD_{i,\rm C} + \mathsmaller{\sum_{i\in\mathcal{W}}}\lambda_{\rm e}\tilde{B}_i\left(e_{i,\rm O}+e_{i,\rm C}\right) \label{Prob_Per_slot_sim_redc_rew_obj}\\
	\st
	~~&r_{i,\rm O} \leq D_{i,\rm O}, \forall i,\label{Prob_Per_slot_sim_cons_rew_aux}\\
	~~&D_{i,\rm C} + r_{i,\rm O} \leq Q_i, \forall i, \label{Prob_Per_slot_sim_cons_rew_data}\\
	~~&\mathsmaller{\sum_{i\in\mathcal{W}}}\tau_i \leq T, \label{Prob_Per_slot_sim_cons_rew_time}\\
	~~& 0 \!\leq \!e_{i,\rm O} \!\leq \!p_i^{\rm max}\tau_i, ~0 \leq\! f_i \!\leq\! f_i^{\rm max}, ~0\leq \tau_i \leq T, \forall i, \label{Prob_Per_slot_sim_cons_rew_pft}
	\end{align}
\end{subequations}
where $D_{i,\rm O}={W\tau_i}\log_2\left(1+\frac{e_{i,\rm O}\gamma_i}{\tau_i}\right)$. Define a subset of WDs $\mathcal{V}=\{i|Q_i-Q_0<0,\forall i\in\mathcal{W}\}$, we have the following interesting result for the optimal solution to \eqref{Prob_Per_slot_sim_rew}.
\begin{Lemma}\label{lem8}
	To achieve the optimum of \eqref{Prob_Per_slot_sim_rew}, we always have $r_{i,\rm O}=p_i=\tau_i=0$ for a WD $i\in\mathcal{V}$, and $r_{i,\rm O} = D_{i,\rm O}$ for a WD $i\in\mathcal{W}\setminus\mathcal{V}$.
\end{Lemma}
\begin{proof}
	For a WD $i\in\mathcal{V}$, we have $Q_i-Q_0<0$. In this case, \eqref{Prob_Per_slot_sim_redc_rew_obj} is a monotonic decreasing function of $r_{i,\rm O}$ and $e_{i,\rm O}$, and thus \eqref{Prob_Per_slot_sim_redc_rew_obj} achieves maximum when $r_{i,\rm O}=p_i=\tau_i=0$. In contrast, for a WD $i\in\mathcal{W}\setminus\mathcal{V}$ where $Q_i-Q_0\geq 0$, \eqref{Prob_Per_slot_sim_redc_rew_obj} monotonically increases with $r_{i,\rm O}$, and we have $r_{i,\rm O} = D_{i,\rm O}$ when \eqref{Prob_Per_slot_sim_rew} achieves optimum.
\end{proof}
By substituting $r_{i,\rm O}=p_i=\tau_i=0,\forall i\in\mathcal{V}$, into \eqref{Prob_Per_slot_sim_redc_rew_obj}, we remove the terms with negative $\left(Q_i-Q_0\right)$ and equivalently transform \eqref{Prob_Per_slot_sim_rew} into a convex optimization problem. As a result, we can use well-established CVX tools such as interior point method to optimally solve \eqref{Prob_Per_slot_sim_rew}. However, as the number of WDs increases, the interior point method exhibits cubic growth of computational complexity in the worst case \cite{Mehrotra1992}. As a result, the CVX-based method may lead to unacceptable computation overhead for online implementation in large-scale MEC systems. As shown in Section \ref{sec6}, the CVX-based method incurs almost 6\% computation overhead to produce an action when the time duration $T=5$ seconds and the number of WDs $K=16$. To resolve this issue, in the following, we propose a BCD-based method that obtains a closed-form solution to \eqref{Prob_Per_slot_sim}. We show in simulations that the proposed BCD-based method achieves almost identical performance as the optimal CVX-based methods, but incurs only 0.04\% of the computation delay. 


%

The BCD-based method (as shown in Algorithm \ref{alg2}) solves \eqref{Prob_Per_slot_sim} by alternately optimizing a) the transmit power $p_i$ and local CPU frequency $\!f_i\!$ and b) the time allocation $\tau_i$, detailed as below. 

\emph{a) Transmit Power and CPU Frequency Control}: By fixing $\tau_i$, $\forall i$, \eqref{Prob_Per_slot_sim} can be decomposed into $K$ parallel subproblems, each of which is in the form of 
\begin{subequations}\label{Prob_Per_slot_sim_sub}
	\small
	\begin{align}
	\underset{\substack{f_i, p_i}} \max~&F(f_i) + G(p_i) \label{Prob_Per_slot_sim_sub_redc_obj}\\
	\st
	~~&D_{i,\rm C} + D_{i,\rm O} \leq Q_i,\label{Prob_Per_slot_sim_cons_rew_data_sub}\\
	~~& 0 \!\leq \!p_i \!\leq \!p_i^{\rm max}, ~0 \leq\! f_i \!\leq\! f_i^{\rm max}, ~0\leq \tau_i \leq T. \label{Prob_Per_slot_sim_cons_rew_pft_sub}
	\end{align}
\end{subequations}
where $F(f_i)=\lambda_{\rm e}\tilde{B}_i\kappa_i\left(f_i\right)^3T\!+\!Q_i\frac{f_iT}{\phi_i}$ and $G(p_i)=\lambda_{\rm e}\tilde{B}_ip_i\tau_i\!+\!\left(Q_i-Q_0\right)\tau_i\log_2\left(1+p_i\gamma_i\right)$. In \eqref{Prob_Per_slot_sim_sub}, we maximize a weighted summation of energy cost and data processing rate at WD$_i$.

\renewcommand{\algorithmicrequire}{\textbf{Initialization:}}
\renewcommand{\algorithmicensure}{\textbf{Output:}}
\begin{algorithm}
	\small
	\caption{Block coordinate descent based method to solve problem \eqref{Prob_Per_slot_sim}} 
	\label{alg2}
	\begin{algorithmic}[1]
		\REQUIRE The iteration index $n=0$ and initial time allocation $\tau_i^0, \forall i$.
		\REPEAT
		\STATE Solve \eqref{Prob_Per_slot_sim_sub} for given $\tau_i^n, \forall i$, and obtain the optimal solution $\{p_i^{n+1}, f_i^{n+1}\}, \forall i$, using \eqref{Opt_pu_fu}.
		\STATE Solve \eqref{Prob_Per_slot_sim_time} for given $\{p_i^{n+1}, f_i^{n+1}\}, \forall i$, and obtain the optimal solution $\tau_i^{n+1}, \forall i$, using \eqref{Opt_t}.
		\STATE Update $n=n+1$.
		\UNTIL The increase of the objective value \eqref{Prob_Per_slot_sim_redc_obj} is below a threshold $\epsilon>0$.
	\end{algorithmic}
\end{algorithm}

To solve \eqref{Prob_Per_slot_sim_sub}, we introduce two auxiliary functions $\mathcal{F}_{p_i}(x) \triangleq \frac{1}{\gamma_i}2^{\frac{1}{W\tau_i}\left(Q_i-\frac{xT}{\phi_i}\right)}-\frac{1}{\gamma_i}$ and $\mathcal{F}_{f_i}(x) \triangleq \frac{\left[Q_i-W\tau_i\log_2\left(1+x\gamma_i\right)\right]\phi_i}{T}$. Intuitively, the feasible CPU frequency $f_i$ and transmit power $p_i$ must satisfy that $f_i\!\in\![0,\! \bar{f}_i^{\rm th}]$ and $p_i\!\in\![0, \bar{p}_i^{\rm th}]$, where $\bar{f}_i^{\rm th}\!=\min\left(f_i^{\rm max}, \mathcal{F}_{f_i}(0)\right)$ and $\bar{p}_i^{\rm th}\!=\!\min\left(p_i^{\rm max},\! \mathcal{F}_{p_i}(0)\right)$, respectively. Due to the nonconvex
constraint \eqref{Prob_Per_slot_sim_cons_rew_data_sub} and time-varying coefficient $(Q_i-Q_0)$ in $G(p_i)$, \eqref{Prob_Per_slot_sim_sub} is generally a non-convex optimization problem. Nevertheless, we derive the closed-form expression of optimal
solution to \eqref{Prob_Per_slot_sim_sub} in the following proposition.
\begin{prop}\label{lem_opt_solution}
	The optimal solution of \eqref{Prob_Per_slot_sim_sub} $\left\{f_i^{\ast},p_i^{\ast}\right\} =$
	\begin{small}
	\begin{subnumcases}{\label{Opt_pu_fu}}
	\left\{\hat{f}_i, \hat{p}_i\right\}, \!& \text{if $Q_i-Q_0\geq 0$ and $\hat{D}_{i,\rm O}^{\rm max}+\hat{D}_{i,\rm C}^{\rm max} \leq Q_i$ and $\tau_i>0$}, \label{Opt_pu_fu_3}\\
	\left\{\hat{f}_i, \mathcal{F}_{p_i}\left(\hat{f}_i\right)\right\}, \!& \text{if $Q_i-Q_0\geq 0$ and $\hat{D}_{i,\rm O}^{\rm max}+\hat{D}_{i,\rm C}^{\rm max} > Q_i$ and $\tilde{B}_t=0$ and $\tau_i>0$}, \label{Opt_pu_fu_2}\\
	\left\{\breve{f}_i, \mathcal{F}_{p_i}\left(\breve{f}_i\right)\right\}, \!& \text{if $Q_i-Q_0\geq 0$ and $\hat{D}_{i,\rm O}^{\rm max}+\hat{D}_{i,\rm C}^{\rm max} > Q_i$ and $\tilde{B}_t<0$ and $\tau_i>0$}, \label{Opt_pu_fu_4}\\
	\left\{\hat{f}_i, 0\right\}, \!& \text{if $Q_i-Q_0<0$ or $\tau_i=0$}.\label{Opt_pu_fu_5}
	\end{subnumcases}
	\end{small}Here, $\hat{f}_i\!=\!\min\left(\sqrt{\frac{-Q_i}{3\lambda_{\rm e}\tilde{B}_i\kappa_i\phi_i}}, \bar{f}_i^{\rm th}\right)$ and $\hat{p}_i=\left[\frac{\left(Q_0-Q_i\right)W}{\lambda_{\rm e}\tilde{B}_i\ln2}\!-\!\frac{1}{\gamma_i}\right]_0^{\bar{p}_i^{\rm th}}$, with $[\cdot]_x^y=\min(\max(\cdot,x),y)$. $\hat{D}_{i,\rm C}^{\rm max}=\frac{\hat{f}_iT}{\phi_i}$ and $\hat{D}_{i,\rm O}^{\rm max}=W\tau_i\log_2\left(1+\hat{p}_i\gamma_i\right)$. $\breve{f}_i=\left[\bar{f}_i\right]_{f_i^{\rm lb}}^{f_i^{\rm ub}}$, with $f_i^{\rm ub} = \hat{f}_i$, $f_i^{\rm lb} = \max\left(0, \mathcal{F}_{f_i}\left(\hat{p}_i\right)\right)$, and $\bar{f}_i\in[0,+\infty)$ is the unique solution of $U^\prime(f_i) = 3\lambda_{\rm e}\kappa_iT\tilde{B}_i\left(f_i\right)^2 - \frac{\lambda_{\rm e}\tilde{B}_iT\ln 2}{W\phi_i\gamma_i}2^{\frac{Q_i}{W\tau_i}-\frac{f_iT}{W\tau_i\phi_i}} + \frac{T}{\phi_i}Q_0 = 0$. In particular, $U^\prime(f_i)$ monotonically decreases with $f_i$, and thus $\bar{f}_i$ can be efficiently obtained via bisection search.
\end{prop}
\begin{proof}
	Please refer to Appendix \ref{App_Opt_solution} for detail.
\end{proof}
The result in \eqref{Opt_pu_fu} shows that the optimal task execution control is closely related to the current available energy $B_i$ (absorbed in $\tilde{B}_i$) and data queue length $\{Q_i, Q_0\}$, detailed as following: 1) The local CPU frequency $\hat{f}_i$ and transmit power $\hat{p}_i$ increase with the battery level $B_i$. 2) A larger data queue length $Q_i$ yields both higher $\hat{f}_i$ and $\hat{p}_i$. 3) The optimal offloading solution exhibits a threshold-based structure: the WD$_i$ performs task offloading only when the local data queue backlog surpasses that at the edge (i.e., $Q_i-Q_0\geq 0$). Otherwise, it only conducts local computation. Besides, the bigger the backlog gap $Q_i-Q_0$, the higher the transmit power $\hat{p}_i$.

\emph{b) Time Allocation for Task Offloading}: Given the $f_i$ and $p_i$ in \eqref{Opt_pu_fu}, $\forall i$, we obtain the optimal time allocation by solving the following LP problem:
\begin{subequations}\label{Prob_Per_slot_sim_time}
	\begin{align}
	\underset{\substack{\tau_i}} \max~&\mathsmaller{\sum_{i=1}^{K}}C_i\tau_i, ~~\st
	~~\mathsmaller{\sum_{i=1}^{K}}\tau_i \leq T, ~0\leq \tau_i\leq \tau_i^{\rm ub},
	\end{align}
\end{subequations}
where $C_i=\left(Q_i-Q_0\right)W\log_2\left(1+\frac{p_ih_i}{N_0}\right) + \lambda_{\rm e}\tilde{B}_ip_i$. $\tau_i^{\rm ub}=\min\left(T, \frac{Q_i-D_{i,\rm C}}{W\log_2\left(1+{p_i\gamma_i}\right)}\right)$ is obtained by absorbing \eqref{Prob_Per_slot_sim_cons_rew_data_sub} into the box constraint $0\leq \tau_i \leq T$. In a special case of $p_i=0$, we have $\tau_i^{\rm ub}=T$. We sort the WDs in an decreasing manner according to $C_i$, i.e., $j\leq i$ if $C_j\geq C_i$. Then, the optimal solution of \eqref{Prob_Per_slot_sim_time} is
\begin{subnumcases}{\tau_i^{\ast} =\label{Opt_t}}
0, \!&\text{if $C_i\leq 0$}, \label{Opt_t0}\\
\min\left(\tau_1^{\rm ub},T\right), \!& \text{if $i=1$ and $C_i> 0$}, \label{Opt_t1}\\
\max\left(\min\left(\tau_i^{\rm ub},T-\mathsmaller{\sum_{j=1}^{i-1}}\tau_j^{\rm ub}\right),0\right), \!& \text{if $1\leq i\leq K$ and $C_i> 0$}. \label{Opt_t2}
\end{subnumcases}

The result in \eqref{Opt_t} shows that the optimal time allocation has a threshold-based structure: the WD$_i$ has non-zero offloading time only when $C_i>0$ and shuts down the offloading circuits otherwise. Since $C_i$ increases with the battery level $B_i$ and the difference of $Q_i-Q_0$, the WD$_i$ is allocated with a large $\tau_i$ when $B_i$ and $Q_i$ are large. Notice that $C_i\leq 0$ when $Q_i-Q_0<0$. The optimal time allocation is align with the optimal control of transmit power in \eqref{Opt_pu_fu}. Overall, the BCD-based LEESE algorithm creates a close-loop control to stabilize both $Q_i$ and $Q_0$. 

\section{Performance Analysis}\label{secV}
In this section, we analyze the performance of the proposed LEESE algorithm. Recall that LEESE removes the energy causality constraint \eqref{EH_caus} when solving (P1). Here, we first show that LEESE can always satisfy \eqref{EH_caus} as long as the battery capacity $\Omega_i$ satisfies a mild condition. Then, we prove that LEESE also meets all the long-term performance requirements, thus producing a feasible solution to (P1). In addition, it achieves a tradeoff between data sensing performance and computation delay by tuning the Lyapunov parameter $V$. 

We first show in the following Lemma \ref{lem_ub} that the data queue length $Q_i(t)$ is bounded above. 
\begin{Lemma}\label{lem_ub}
	Given an initial data queue satisfying $Q_i(0)\in\left[0,Q^{\rm max}\right]$, the data queue length at the WD$_i$ satisfies that $0\leq Q_i(t)\leq Q^{\rm max}$, for $t=0,1,2,\cdots$, where $Q^{\rm max} = V+r_{\rm max}$.
\end{Lemma}
\begin{proof}
	We prove this lemma by induction. In particular, the base case is that $0\leq Q_i(t)\leq Q^{\rm max}$ holds initially at $t=0$. For the induction step, we suppose that $0\leq Q_i(t)\leq Q^{\rm max}$ holds for time slot $t$. In the following, we prove that $0\leq Q_i(t+1)\leq Q^{\rm max}$ is true by considering two cases: a) When the WD$_i$ collects zero-bit data in time slot $t$, it is straightforward that $Q_i(t+1)\leq Q_i(t) \leq V+r_{\rm max}$. b) When the WD$_i$ collects $r_i(t)$-bit data in time slot $t$, we have $r_i(t)=r_{\rm max}$ and $V+\lambda_{\rm e}(B_i(t)\!-\!\Omega_i)e_{\rm unit}^{\rm col}-Q_i(t)\geq 0$ according to \eqref{Opt_r}. Since $B_i(t)\!-\!\Omega_i\leq 0$ holds for all $t$, we have $Q_i(t) \leq V$ and thus $Q_i(t+1)\leq Q_i(t) + r_{\rm max} \leq V+r_{\rm max}$. By the induction principle, it follows that $0\!\leq\! Q_i(t)\!\leq\! Q^{\rm max}$ is true for all $t$, which completes the proof.
\end{proof}

Lemma \ref{lem_ub} shows that although we assume in this paper infinite data buffer capacity for simplicity of analysis, the data queue length is in fact bounded when implementing LEESE. In the following Proposition \ref{lem3}, we derive a sufficient condition on $\Omega_i$ to safely remove the energy causality constraint \eqref{EH_caus} when implementing LEESE to solve (P1).

\begin{prop}\label{lem3}
	The proposed LEESE algorithm respects the energy causality constraint \eqref{EH_caus} in every time slot if 
	\begin{equation}\label{Bat_thre}
		\Omega_i\!\geq\!\max\left(\!\frac{V}{\lambda_{\rm e}e_{\rm unit}^{\rm col}}+\lambda_{\rm e}e_i^{\rm max},\bar{\Omega}_i\!\right)+\lambda_{\rm e}p_0^{\rm max}T, \forall i,
	\end{equation}
	where $\bar{\Omega}_i$ is the unique solution of equation $H(\Omega_i) = 0$ and can be efficiently found via bisection search, where
	$H(\Omega_i) = \lambda_{\rm e}\kappa_i\left(\sqrt{\frac{V+r_{\rm max}}{3\lambda_{\rm e}\left(\Omega_i-\lambda_{\rm e}e_i^{\rm max}\right)\kappa_i\phi_i}}\right)^3T \!+\! \frac{\left(V+r_{\rm max}\right)W}{\lambda_{\rm e}\left(\Omega_i-\lambda_{\rm e}e_i^{\rm max}\right)\ln 2}T\!-\!B_{\rm min}.$
\end{prop}
\begin{proof}
	Please refer to Appendix \ref{app1} for detail.
\end{proof}
Proposition \ref{lem3} shows that the energy causality constraint becomes implicit when implementing the LEESE algorithm to solve (P1) given sufficient battery capacity for each WD. In practice, the threshold is $\Omega_i\geq 15.3$ Joules using the parameters in Section IV, which is satisfied for a common commercial battery with tens of thousand Joules capacity \cite{Muenzel2015}. As shown in Proposition \ref{lem3}, larger $V$ leads to larger battery capacity. Since the data sensing rate increases with $V$ (as shown later in Proposition \ref{lem4}), LEESE provides a tradeoff between the battery capacity and achievable data sensing performance. That is, by increasing the control parameter $V$, we achieves higher data sensing rate at cost of more expensive battery hardware with larger capacity. 

In the following, we show that the online solution produced by LEESE algorithm satisfies all the long-term constraints in \eqref{DataQ_stab} and \eqref{Bud_cons} and achieves a $\left[O(1/V),O(V)\right]$ sensing rate-delay performance tradeoff. To start with, we introduce the following auxiliary problem:
\begin{subequations}\label{Primal_prob_mod}
	\begin{align}
	(\text{P2})~~	\underset{\substack{r_i(t), p_i(t),f_i(t), \\ p_0(t), f_0(t),\forall i, t}} \max~&\lim_{N\to+\infty}\frac{1}{N}\mathsmaller{\sum_{t=0}^{N-1}\sum_{i=1}^{K}}w_ir_i(t)\\
	\st
	~~& \eqref{ledge_off_cons}, \eqref{DataQ_stab}, \eqref{T_allo_cons}, \eqref{Bud_cons}, \eqref{Prob_const2},\\
	~~& \lim_{N\to+\infty}\frac{1}{N}\mathsmaller{\sum_{t=0}^{N-1}}\mathbb{E}\left[e_i(t)-e_{i,\rm h}(t)\right] \leq 0, \forall i. \label{Energy_lonterm}
	\end{align}
\end{subequations}
Compared to (P1), (P2) uses the long-term energy constraint \eqref{Energy_lonterm} in substitution for the energy causality constraint \eqref{EH_caus} in (P1). Denote the optimal value of (P1) and (P2) as $\bar{R}_{\rm P1}^\ast$ and $\bar{R}_{\rm P2}^\ast$, respectively. We prove in the following lemma that $\bar{R}_{\rm P1}^\ast\leq \bar{R}_{\rm P2}^\ast$.
\begin{Lemma}\label{lem1}
	(P2) is a relaxation of (P1), i.e., $\bar{R}_{\rm P1}^\ast\leq \bar{R}_{\rm P2}^\ast$.
\end{Lemma}
\begin{proof}
	For any feasible solution of (P1), the battery dynamics of WD$_i$, $\forall i$, must satisfy
	\begin{equation}
	B_i(t+1) \leq B_i(t) - \lambda_{\rm e}e_i(t) + \lambda_{\rm e}e_{i,\rm h}(t), t = 0, \cdots, N-1.
	\end{equation}
	By summing up both sides over $t = 0, \cdots, N-1$, taking the expectation, diving both sides by $\lambda_{\rm e}N$ and letting $N$ go to infinity, we have
	\begin{equation}
	\lim_{N\!\to\!+\!\infty}\!\frac{1}{\lambda_{\rm e}N}\mathbb{E}\left[B_i(N)\right]\!\leq\! \lim_{N\!\to\!+\!\infty}\!\!\frac{1}{\lambda_{\rm e}N}\mathbb{E}\left[B_i(0)\right]\! -\! \lim_{N\!\to\!+\!\infty}\!\frac{1}{N}\!\sum_{t=0}^{N-1}\!\mathbb{E}\left[e_i(t)-e_{i,\rm h}(t)\right].
	\end{equation}
	Since $B_i(t)\leq\Omega_i<+\infty$, we have $\lim_{N\to+\infty}\frac{1}{\lambda_{\rm e}N}\mathbb{E}\left[B_i(N)\right]=\lim_{\lambda_{\rm e}N\to+\infty}\frac{1}{N}\mathbb{E}\left[B_i(0)\right]=0$, i.e., the feasible solution of (P1) satisfies the constraint \eqref{Energy_lonterm} in (P2). Hence, (P2) is a relaxation of (P1), and thus $\bar{R}_{\rm P1}^\ast\leq \bar{R}_{\rm P2}^\ast$.
\end{proof}
The considered fading channels are i.i.d. across different time slots. We define a class of stationary and randomized policy called $\omega$-only policy \cite{Neely2010}, which makes control decisions only based on the current channel state (and independent to the queue backlog state $\Theta$). We assume (P2) is feasible and satisfies the following Slater condition.  
\begin{Assumption}
	There exists a constant $\epsilon\!>\!0$ and $\varphi(\epsilon)\!\!\leq \!\!\bar{R}_{\rm P2}^\ast$ and an $\omega$-only policy $\Gamma$ satisfying that
	\begin{subequations}\label{SLT}
		\setlength{\abovedisplayskip}{-5pt}
		\setlength{\belowdisplayskip}{-5pt}
		\begin{align}
		&\mathbb{E}\left[\mathsmaller{\sum_{i=1}^{K}}r_i^{\Gamma}(t)\right] = \varphi(\epsilon),~
		\mathbb{E}\left[e_0^{\Gamma}(t)\right] \leq  e_0^{\rm th}-\epsilon,~ \mathbb{E}\left[\mathsmaller{\sum_{i=1}^{K}}D_{i,\rm{O}}^{\Gamma}(t)\right]\leq\mathbb{E}\left[D_{0,\rm{C}}^{\Gamma}(t)\right]-\epsilon,
		\\
		&\mathbb{E}\left[e_i^{\Gamma}(t)-e_{i,\rm h}^{\Gamma}(t)\right] \leq -\epsilon,~
		\mathbb{E}\left[r_i^{\Gamma}(t)\right] \leq \mathbb{E}\left[D_{i,\rm{O}}^{\Gamma}(t)+D_{i,\rm{C}}^{\Gamma}(t)\right] -\epsilon, \forall i.
		\end{align}
	\end{subequations}
\end{Assumption}

In the following Proposition \ref{lem4}, we prove that when using convex optimization based method or BCD-based method to solve subproblem \eqref{Prob_Per_slot_sim}, LEESE keeps stability of system data queues and achieves an asymptotically optimal solution to the primary problem (P1).

\begin{prop}\label{lem4}
	Suppose that for any $s(t)$ in time slot $t$, LEESE produces a solution to subproblem \eqref{Prob_Per_slot_sim} with limited optimality gap $\upsilon\geq 0$, and \eqref{Bat_thre} is satisfied. Then, when implementing the proposed LEESE algorithm to solve (P1), we have that:
	\begin{itemize}
		\item [a)] The achievable long-term average data sensing rate, denoted as $\bar{R}_{\Psi}$, has a lower bound
		\begin{equation}\label{Asy_opt_R}
		\begin{split}
		\bar{R}_{\Psi}\geq \bar{R}_{\rm P1}^\ast - \frac{C+ \upsilon}{V},
		\end{split}
		\end{equation}
	where $C$ is a constant defined in \eqref{Dvalue}.
		\item[b)] The virtual power deficit queue $Z(t)$ is strongly stable, and the long-term average power constraint \eqref{Bud_cons} is satisfied.
		\item[c)] The data queue stability \eqref{DataQ_stab} is guaranteed. The average data queue length at the HAP and WD$_i$, $\forall i$, satisfy that
		\begin{subequations}\label{lem5_eq2}
			\begin{align}
			&\lim_{N\!\to\!+\!\infty}\frac{1}{N}\mathsmaller{\sum_{t=0}^{N}}\mathbb{E}\left[Q_0(t)\right] \leq \frac{C+ \upsilon+V\left(\bar{R}_{\rm P1}^\ast-\varphi(\epsilon)\right)}{\epsilon}<+\infty,\\
			&\lim_{N\!\to\!+\!\infty}\frac{1}{N}\mathsmaller{\sum_{t=0}^{N}}\mathbb{E}\left[Q_i(t)\right] \leq \frac{C+ \upsilon+V\left(\bar{R}_{\rm P1}^\ast-\varphi(\epsilon)\right)}{\epsilon}<+\infty.
			\end{align}
		\end{subequations}
	\end{itemize}
\end{prop}
\begin{proof}
	Please refer to Appendix \ref{app2} for detail.
\end{proof}

As shown in Proposition \ref{lem3} and \ref{lem4}, LEESE can produce a feasible solution to problem (P1) and achieves an $\left[O(1/V),O(V)\right]$ sensing rate-delay tradeoff given a limited optimality gap $\upsilon$ and a large enough battery capacity $\Omega_i$. Specifically, as $V$ increases, LEESE can improve the data sensing rate at rate of $O(1/V)$, but at the price of longer data queue length (processing delay) increasing at rate of $O(V)$. In Section \ref{sec6}, we will investigate the impact of $V$ on the long-term sensing performance and show that LEESE achieves a negligibly small $\upsilon$ for the subproblem \eqref{Prob_Per_slot_sim}.

\section{Numerical Results}\label{sec6}
{\color{black}In this section, we evaluate the performance of the proposed LEESE algorithm via simulations. All simulations are conducted by Python 3.8 and on a computer with 16 GB RAM and Inter Core i7-6700 3.4 GHz CPU.} In all simulations, we adopt the parameters of the Powercast TX91501-3W transmitter with maximum transmit power $p_0^{\rm max}=3$ Watt and carrier frequency $f_{\rm c}^{\rm p}=915$ MHz as the wireless power transmitter at the HAP. Unless otherwise stated, we consider $K=8$ WDs with identical maximum transmit power $p_i^{\rm max}=5$ dBm, $\forall i$, and communication carrier frequency $f_{\rm c}^{\rm I}=2.4$ GHz. We set the distance between the HAP and WD$_i$ as $d_i = 2 + \frac{8}{K-1}(i-1)$ in meters, for $i=1,\cdots, K$. We set energy harvesting parameters $a_{1,i}=2.463$, $a_{2,i}=1.635$ and $a_{3,i}=0.826$, $\forall i$ \cite{Chen2017}. For both WPT and data communication, we model the channels as Rayleigh fading channels, where the corresponding channel gain is $h_i^x(t)=\varsigma^x(t)G_{\rm A}\left(\frac{3\times10^8}{4\pi f_{\rm c}^xd_i}\right)^{\sigma}$, where $x\in\{\rm p, I\}$ indicates the channel for WPT and data communication, respectively. $\varsigma^x(t)$ is an independent exponential random variable of unit mean. $G_{\rm A}=4.11$ captures the total antenna gain. $\sigma\geq 2$ denotes the path-loss exponent. Unless otherwise stated, we set $\sigma=2.4$. Likewise, we set equal $w_i=1$, $k_i=10^{-26}$, $\phi_i=1000$, and $f_i^{\rm max}=16$ MHz, for all $i=1,\cdots,K$. In this case, all the WDs have the same battery capacity (denoted as $\Omega$) according to Proposition \ref{lem3}. We initialize all the data queue length to 0, i.e., $Q_0(0)=Q_i(0)=0, \forall i$, and initialize full battery level at all WDs, i.e., $B_i=\Omega, \forall i$. Besides, we set the time-average energy budget at the HAP $e_0^{\rm th}=8$ Joules, and the maximum sensing data size $r_{\rm max}=512$ Kbits. The simulation length is set to $N=5\times10^4$ time slots. The other parameters used in simulation are listed in Table \ref{tab1}. 

\vspace{-10pt}
\begin{table}
	\centering
	\caption{Simulation Parameters}
	\label{tab1}
	\begin{tabular}{|c|c|c|}
		\hline
		$W = 30$ KHz  & $e_{\rm unit}^{\rm col}=10^{-9}$ Joules/bit & $T = 5$ second \\
		\hline
		$k_0=10^{-26}$  & $f_0^{\rm max} = 2$ GHz  &$\phi_0 = 1000$ cycles/bit  \\
		\hline
		$N_0 = 10^{-10}$ Watt  & $B_{\rm min} = 10^{-3}$ Joules & $V=32\times10^5$ \\
		\hline	
	\end{tabular}
\vspace{-0pt}
\end{table}

\subsection{Asymptotic Optimality of LEESE}
We first investigate the performance of the proposed BCD-based method in solving the sub-problem \eqref{Prob_Per_slot_sim}. For comparison, we consider the CVX-based method as the optimal benchmark. In particular, we apply the Python-embedded CVXPY solver \cite{Diamond2016} and the BCD-based method to solve \eqref{Prob_Per_slot_sim}, respectively, and record the obtained results over $500$ time slots in Fig. \ref{L0_0}. The result shows that the BCD-based method achieves almost identical performance as the optimal solution when solving \eqref{Prob_Per_slot_sim}, and thus achieves a negligibly small optimality gap $\upsilon$. This result confirms the asymptotic optimality of LEESE in Proposition \ref{lem4}.\
\begin{figure}[tbp]
	\centering
	\includegraphics[scale=0.58]{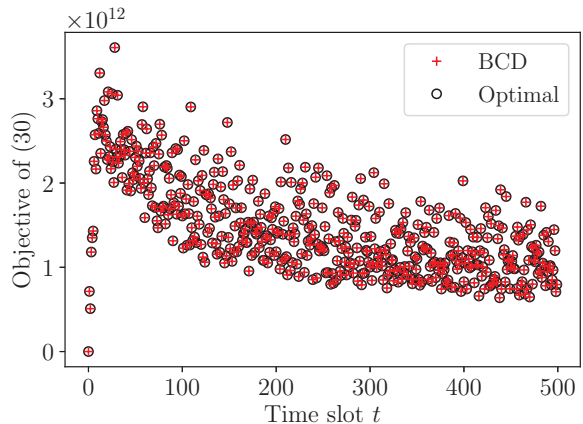}
	\captionsetup{font=footnotesize}
	\caption{Performance of the BCD-based method in solving per-slot subproblem \eqref{Prob_Per_slot_sim}.}
	\label{L0_0}
\end{figure}

Then, we investigate in Fig. \ref{L1_0} the impact of parameter $V$ on the performance of LEESE. For convenience, we denote $\mathbb{E}[\bar{Q}_i]$ as the mean value of the long-term average data queue length $\bar{Q}_i$ over $K$ WDs. It displays that the data sensing rate $\bar{R}$ and data queue length at the HAP $\bar{Q}_{0}$ grow with $V$ and gradually saturate when $V$ is large (i.e., $V\geq 32\times10^5$). However, $\mathbb{E}[\bar{Q}_i]$ and $\Omega$ grow rapidly with $V$ especially when $V> 32\times 10^5$. These results agree with the theoretical analysis in section \ref{secV}, where a larger $V$ produces a higher $\bar{R}$ and longer data queues. In the following simulations, we set $V=32\times10^5$, whereby LEESE achieves near-optimal data sensing rate with small data queue length and requires small battery capacity.

\begin{figure}
	\begin{minipage}{0.48\linewidth}
		\centerline{\includegraphics[scale=0.32]{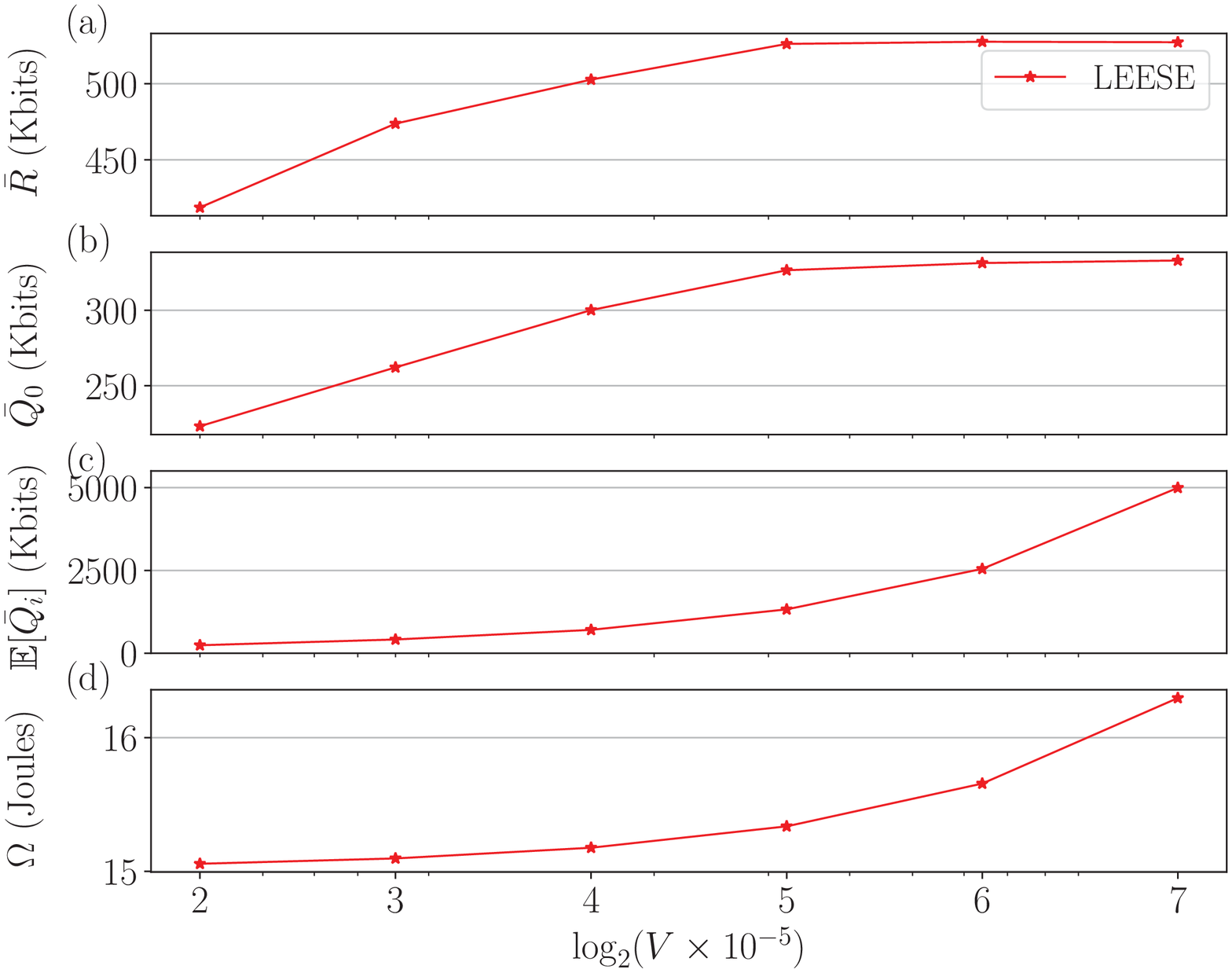}}
		\captionsetup{font=footnotesize}
		\caption{System performance of LEESE under different $V$.}
		\label{L1_0}
	\end{minipage}
	\hfill
	\begin{minipage}{.48\linewidth}
		\centerline{\includegraphics[scale=0.32]{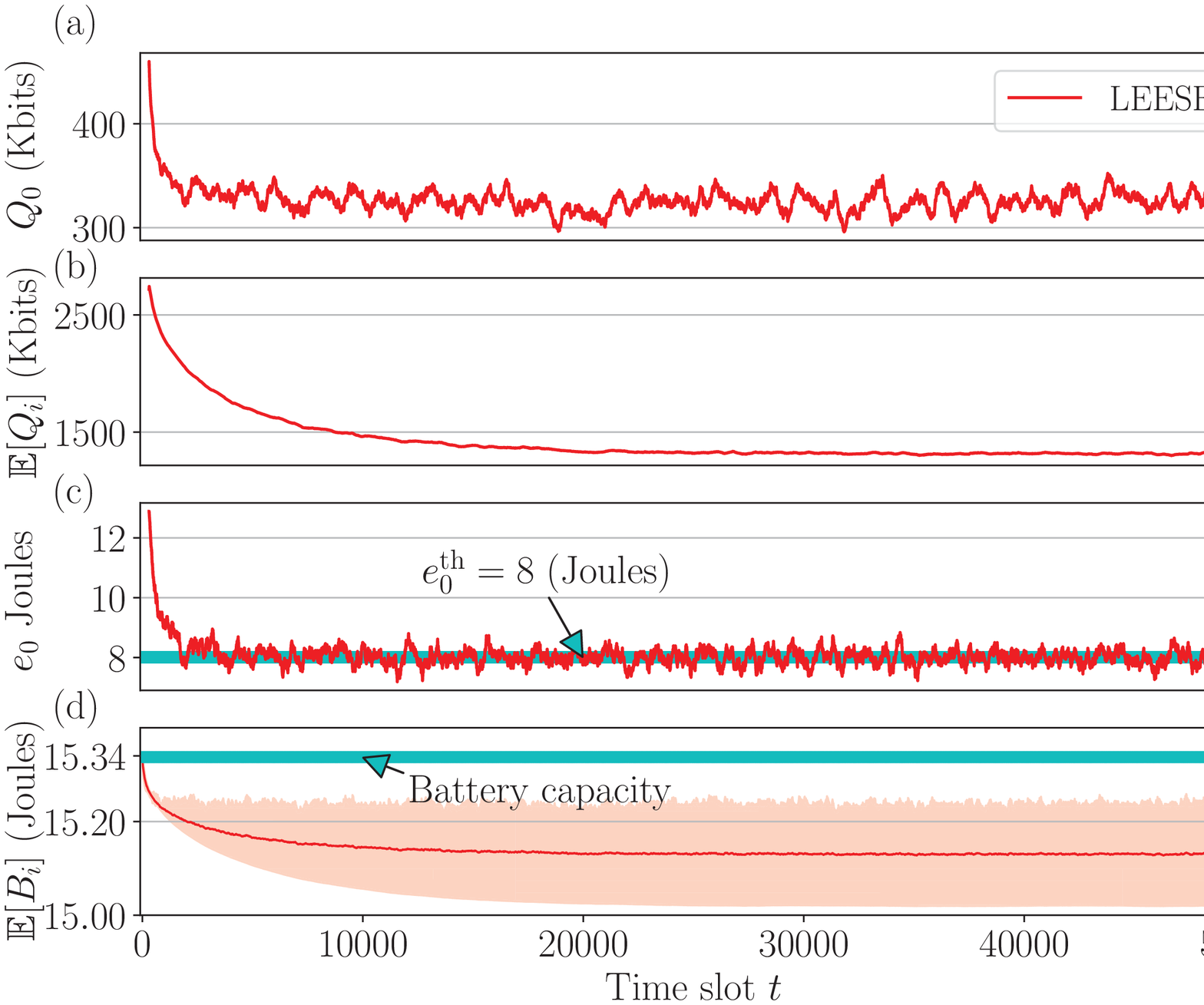}}
		\captionsetup{font=footnotesize}
		\caption{Feasibility of PLySE and Benchmark methods.}
		\label{L2_0}
	\end{minipage}
\end{figure}


\subsection{Feasibility of LEESE}
We investigates the feasibility of the proposed LEESE algorithm to the problem (P1). For convenience, we denote $\mathbb{E}[Q_i]$ and $\mathbb{E}[B_i]$ as the per-slot average data queue length and per-slot average battery level over $K$ WDs. In Fig. \ref{L2_0}(a)-(c), we plot the moving average of the data queue length $Q_0$ and $\mathbb{E}[Q_i]$ as well as the edge energy consumption $e_0$ over the last 400 time slots. The results show that LEESE stabilizes the data queues $Q_0$ and $Q_i$, and satisfies the long-term average energy consumption constraint at the HAP. We also plot in Fig. \ref{L2_0}(d) the average battery level $\mathbb{E}[B_i]$ as time proceeds (the red line), and the maximum and minimum $B_i$ in each time slot (the light red shadow). We observe that the battery levels of all WDs are all in the range of $[0,\Omega]$ over time, which means that LEESE satisfies the energy causality constraints even when it is removed when solving for the online control solutions. All these results verify the feasibility of the proposed LEESE algorithm.


\subsection{Performance Comparison Under Various System Parameters}
To verify the effectiveness of the proposed LEESE algorithm, we consider the following four representative benchmark methods for comparison:

1) CVX-based LEESE (CVX-LEESE): It minimizes the upper bound of drift-plus-penalty function in \eqref{LyaDriPlusPen} and obtains the optimal solution to the per-slot problem \eqref{Prob_Per_slot} using convex optimization algorithms, e.g., the Python-embedded CVXPY solver.

2) Local computing only (LCO): All WDs compute all tasks locally rather than task offloading. Specifically, LCO computes $p_0^\ast$, $f_0^\ast$ and $r_i^\ast$ similar to LEESE (i.e., using \eqref{Opt_p0}, \eqref{Opt_fs} and \eqref{Opt_r}, respectively), while obtains $f_{i}^{\ast}$ by solving \eqref{Prob_Per_slot_sim} with $p_{i}^{\ast}=\tau_i^\ast=0$.

3) Equal offloading time (EqOT): EqOT computes $p_0^\ast$, $f_0^\ast$ and $r_i^\ast$ similar to LEESE, while obtains $f_{i}^{\ast}$ and $p_{i}^{\ast}$ by solving \eqref{Prob_Per_slot_sim} with $\tau_i^\ast=\frac{T}{K}, \forall i$.
 
4) Myopic edge processing (MyopicEdge): The MyopicEdge method replaces the long-term average power constraint \eqref{Bud_cons} by $N$ per-slot energy cost constraints, i.e., $k_0f_0(t)^3T+p_0(t)T\leq te_0^{\rm th}-\sum_{n=1}^{t-1}e_0(n)$, for $t=1,\cdots, N$. In time slot $t$, it computes $r_i^\ast$, $f_{i}^{\ast}$, $p_{i}^{\ast}$ and $\tau_i^\ast$ similar to LEESE, while obtains $p_0^\ast$ and $f_0^\ast$ by solving the following problem:
\begin{subequations}\label{Prob_Myopic}
	\begin{align}
	\underset{\substack{p_0, f_0}}\max~&Q_0{f_0T}/{\Phi_0}+\mathsmaller{\sum_{i=1}^{K}}\lambda_{\rm e}\left(\Omega_i-B_i\right)e_{i,\rm h}\\
	\st
	&~~\eqref{Subprob_p0_constr}, \eqref{Subprob_fs_const},\\ &~~k_0f_0^3T+p_0T\leq te_0^{\rm th}-\mathsmaller{\sum_{n=1}^{t-1}}e_0(n),
	\end{align}
\end{subequations}
where we drop the time index ``t'' for concision. By doing so, the MyopicEdge method tends to exhaust all the battery energy available for optimal per-slot performance. \eqref{Prob_Myopic} is a convex optimization problem and can be solved by off-the-shelf convex optimization tools like interior point method. It is worth noting that EqOT and MyopicEdge optimize the decisions on either wireless power transfer or data processing, while the proposed LEESE algorithm achieves a joint optimization on both of them. To distinguish from the CVX-LEESE, we denote the proposed LEESE algorithm based on BCD as BCD-LEESE in the following results.

\begin{figure}
	\begin{minipage}{0.48\linewidth}
		\centerline{\includegraphics[scale=0.54]{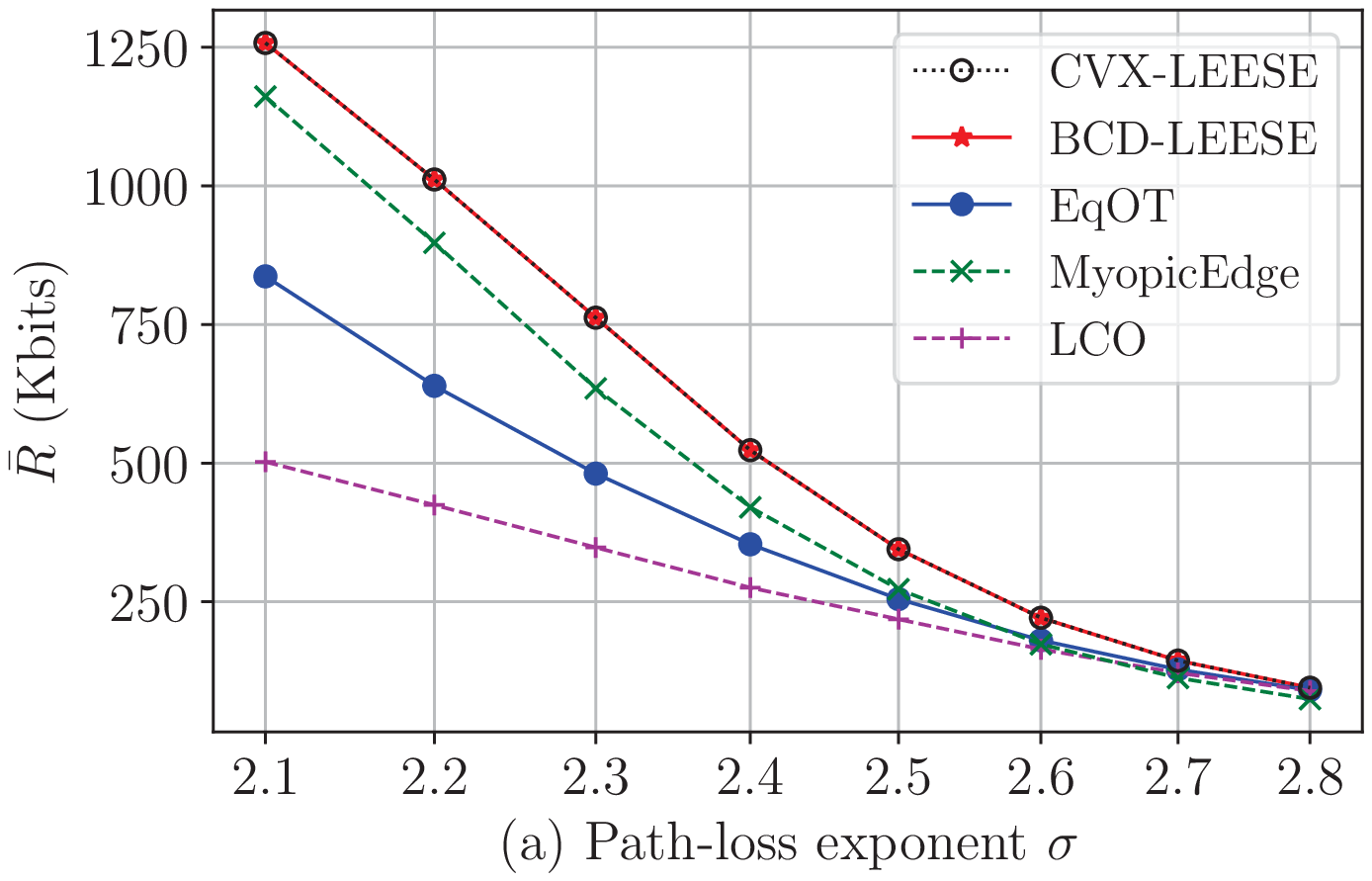}}
	\end{minipage}
	\hfill
	\begin{minipage}{.48\linewidth}
		\centerline{\includegraphics[scale=0.54]{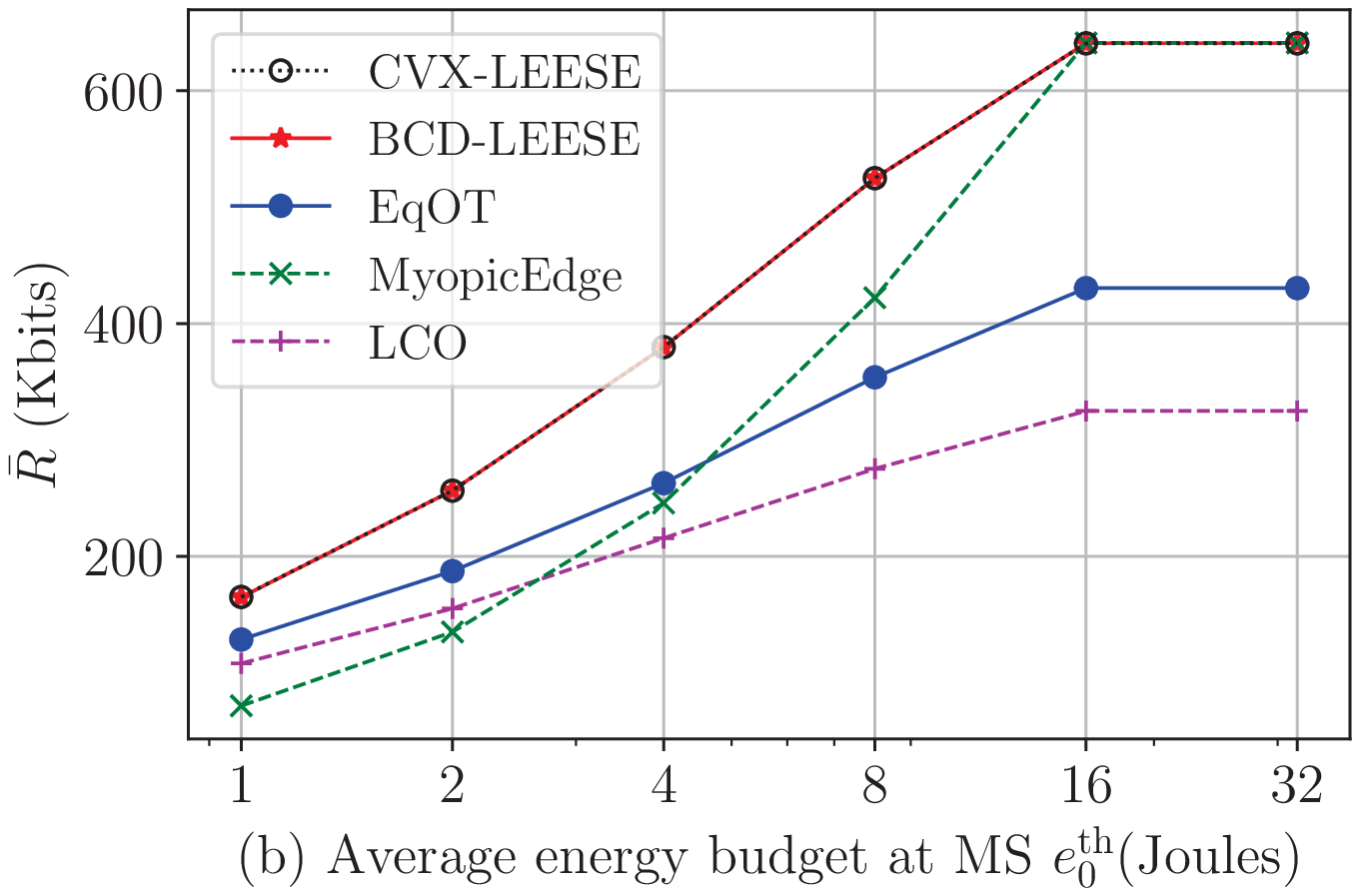}}
	\end{minipage}
	\vfill
	\begin{minipage}{0.48\linewidth}
		\centerline{\includegraphics[scale=0.54]{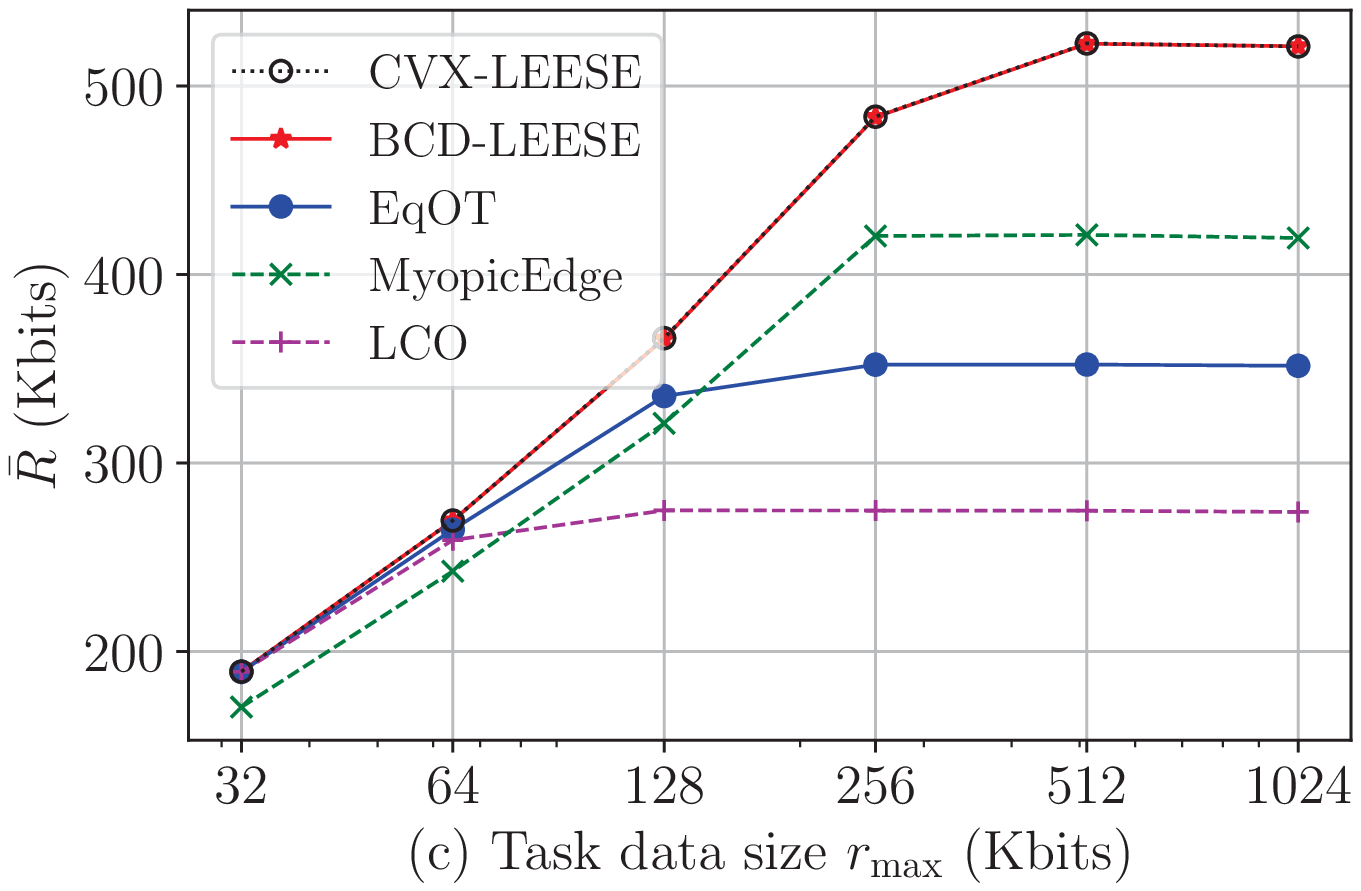}}
	\end{minipage}
	\hfill
	\begin{minipage}{0.48\linewidth}
		\centerline{\includegraphics[scale=0.54]{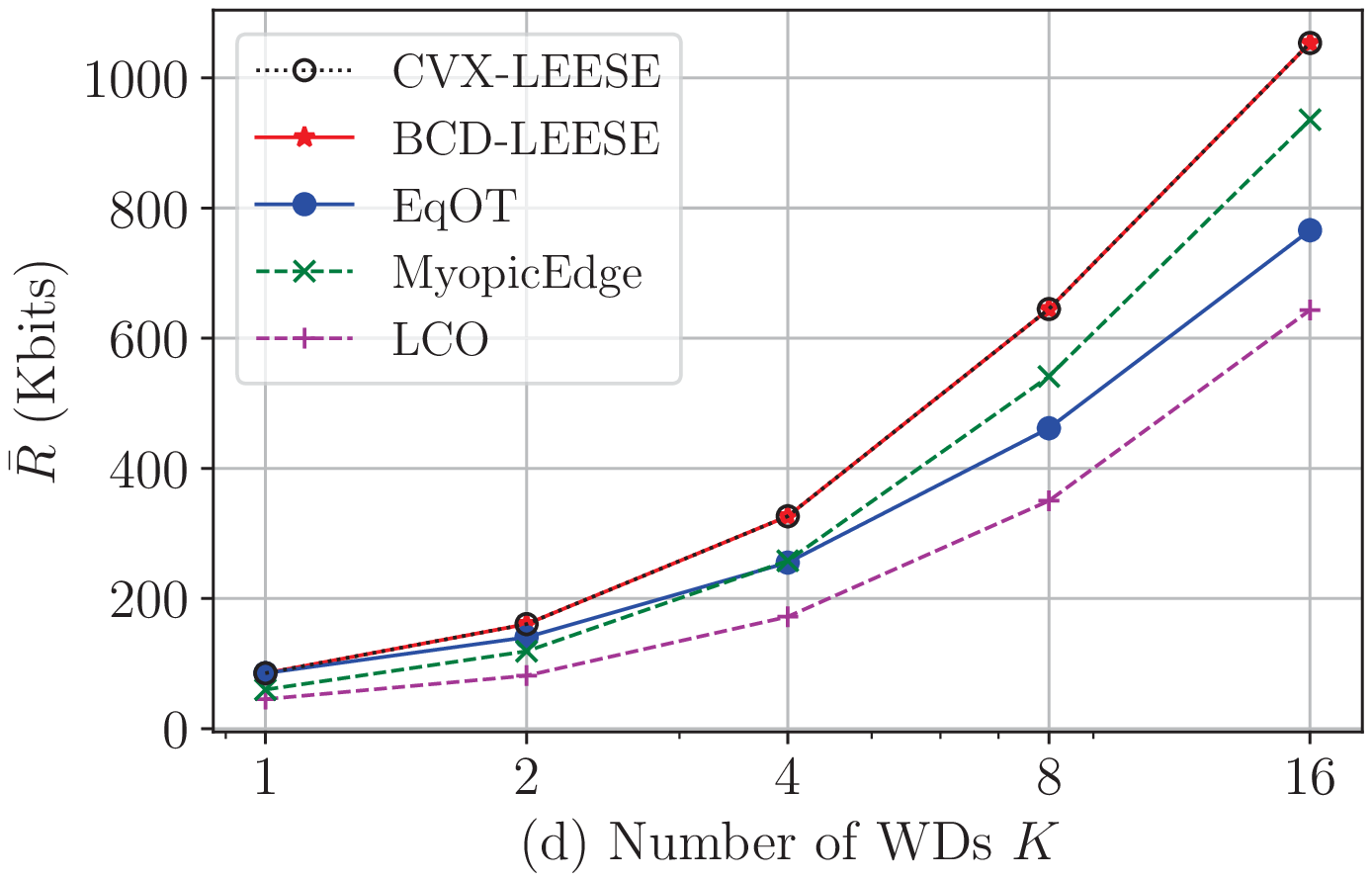}}
	\end{minipage}
	\captionsetup{font=footnotesize}
	\caption{{Long-term average data sensing rate $\bar{R}$ versus: (a) path-loss exponent between the HAP and WDs $\sigma$; (b) energy budget $e_0^{\rm th}$ at the HAP; (c) maximum sensing data size $r_{\rm max}$; and (d) number of WDs $K$.}}
	\label{L3}
\end{figure}

We first compare in Fig. \ref{L3}(a) the performance of BCD-LEESE and benchmark methods under different path-loss exponent $\sigma$. As $\sigma$ increases, all the methods observes a rapid performance degradation. This is because that a larger $\sigma$ leads to a more severe signal attenuation during computation offloading and less received energy at WDs. Since a WD prefers to local computing when experiencing deep fading channel, all the methods eventually achieves a similar $\bar{R}$ as $\sigma$ grows. We also observe that CVX- and BCD-based LEESE achieve the same data sensing rate for all considered $\sigma$'s, and show a significant superiority over the other three benchmark methods. In particular, BCD-LEESE offers 76.9\%, 36.2\%, 21.9\% higher data sensing performance in average than LCO, EqOT, and MyopicEdge, respectively.

In Fig. \ref{L3}(b), we plot the long-term average data sensing rate when the edge energy budget $e_0^{\rm th}$ varies. As shown in the figure, BCD-LEESE achieves similar performance as CVX-LEESE, and greatly outperform the other three benchmarks for all considered $e_0^{\rm th}$'s. The average data sensing rate of all the methods increases with $e_0^{\rm th}$ and finally saturate due to the limited system resources (e.g., transmit power and local CPU frequency). As $e_0^{\rm th}$ increases, BCD-based LEESE enlarges its performance gaps to LCO and EqOT. On the contrary, the data sensing performance of MyopicEdge is comparable to BCD-based LEESE when $e_0^{\rm th}\geq 16$ Joules, but dramatically deteriorates when $e_0^{\rm th}$ decreases. This is because that a sufficiently large $e_0^{\rm th}$ covers the energy waste during WPT under MyopicEdge. 


We also show in Fig. \ref{L3}(c) the impact of task data size $r_{\rm max}$ on the average data sensing performance. As $r_{\rm max}$ grows, all the methods observe increasing data sensing rate and finally achieve a saturated $\bar{R}$ due to limited energy budget. Compared to the other four methods, MyopicEdge is more sensitive to sensing data size. Specifically, MyopicEdge yields the worst data sensing performance when $r_{\rm max}\leq 64$ Kbits, while outperforms LCO and EqOT when $r_{\rm max}\geq 256$ Kbits. This is because that MyopicEdge performs a channel-oblivious WPT process. For a small $r_{\rm max}$, it wastes more energy during WPT especially with the non-linear energy harvesting efficiency. For all considered $r_{\rm max}$'s, BCD- and CVX-based LEESE methods are on top of each their and offer substantial performance gain over the other three benchmark methods. 
%

In Fig. \ref{L3}(d), we further evaluate the data sensing performance under different number of WDs $K$. For each $K$, we generate $d_i$ ($i=1\cdots,K$) independently from a truncated Gaussian distribution, i.e., $d_i=\min\left(\max\left(X,2\right),10\right)$, where $X\sim\mathcal{N}(\bar{d},\delta_d)$ is a Gaussian random variable with average WD-HAP distance $\bar{d}=5$ m and standard deviation of placement spread $\delta_d=3$ m. Each point in Fig. \ref{L3}(d) is an average performance over 20 independent placements of the WDs. Intuitively, a larger number of WDs produces a higher energy efficiency. For all the methods, we see that $\bar{R}$ grows with $K$. In a special case of $K=1$, the two LEESE algorithms and EqOT yield the same time allocation decision (i.e., $\tau_1=T$), and thus achieving the equal data sensing rate. Besides, as $K$ increases, we observe a crossover between EqOT and MyopicEdge. This result implies that an energy efficient WPT is more favorable than a deliberated time allocation when $K$ is small, whereas the reverse is the case when $K$ is large. Nevertheless, BCD-LEESE and CVX-LEESE dominate the three benchmark methods for all considered $K$'s in terms of average data sensing rate.

\begin{table}[]
	\centering
	\caption{Average computation time per slot (second)}
	\label{tab2}
	\begin{tabular}{|c|l|c|c|c|c|c|c|}
		\hline
		\multicolumn{2}{|c|}{Number of WDs $K$}           & 1 & 2 & 4 & 8 & 16 & 32\\ \hline
		\multicolumn{2}{|c|}{BCD-LEESE} & $2.0\times10^{-4}$  & $3.1\times10^{-4}$  & $5.2\times10^{-4}$  & $9.8\times10^{-4}$  & $1.8\times10^{-3}$ & $3.6\times10^{-3}$ \\ \hline
		\multicolumn{2}{|c|}{CVX-LEESE} & $7.8\times10^{-2}$  & $8.9\times10^{-2}$  & $0.12$  & $0.18$  & $0.32$ & $0.59$  \\ \hline
	\end{tabular}
	\vspace{-2pt}
\end{table}
As shown in Fig. \ref{L3}, BCD-LEESE algorithm achieves almost identical performance to CVX-LEESE for all considered system parameters. These results confirm the effectiveness of BCD-LEESE in the online design of the considered system. To further show the efficiency of BCD-LEESE, we apply BCD-LEESE and CVX-LEESE to solve the per-slot problem \eqref{Prob_Per_slot} over 5000 time slots, and record in Table \ref{tab2} the average computation time per slot under various number of WDs $K$. The results show that BCD-LEESE takes at most $3.6\times10^{-3}$ seconds to make an online decision in all the considered $K$'s, while CVX-LEESE generates acceptable computation time only when $K$ is very small, e.g., $K=1$ and 2. In general, the channel coherence time of an indoor IoT system is no more than several seconds. The costly computation overhead makes CVX-LEESE incompetent to achieve real-time control in a practical MEC system with large number of WDs. In contrast, the proposed BCD-LEESE algorithm incurs almost neglected latency overhead, e.g., only about 0.04\% overhead when the time duration $T\!=\!5$ seconds for $K\!=\!16$, and thus is particularly suitable for real-time implementation of large-scale MEC systems.

\section{Conclusion}\label{sec7}
In this paper, we proposed an energy-efficient online control policy for a wireless powered MEC system. We formulated a multi-stage stochastic optimization problem that maximizes the long-term average sensing rate of all WDs under system data queue stability and long-term average power constraints at the HAP. The online design requires jointly control decisions on wireless power transfer, data sensing and processing. To solve the problem, we developed an online algorithm named LEESE based on the perturbed Lyapunov optimization technique. {\color{black}We proved that the proposed LEESE algorithm produces a feasible solution to the target problem, and achieves an $[O(1/V),O(V)]$ sensing rate-delay performance tradeoff.} Compared to other representative benchmark methods, numerical results demonstrated that the proposed LEESE algorithm achieves more than 21.9\% higher sensing rate and consumes only sub-millisecond computation time.  

\appendices
\section{Proof of Lemma \ref{lem_convex}}\label{app_convex}
By taking the first-order and second-order derivatives of $e_{i,\rm h}(t)$ with respect to $p_0(t)$, we have
\begin{equation}
	\small
	\frac{\partial e_{i,\rm h}(t)}{\partial p_0(t)} = \frac{h_i(t)(a_{1,i}a_{3,i}-a_{2,i})}{(p_0(t)h_i(t)+a_{3,i})^2},~\frac{\partial^2 e_{i,\rm h}(t)}{\partial p_0(t)^2} = \frac{2(h_i(t))^2(a_{2,i}-a_{1,i}a_{3,i})}{(p_0(t)h_i(t)+a_{3,i})^3}.
\end{equation}
Since the harvested energy monotonically increases with the input power, we have $\frac{\partial e_{i,\rm h}(t)}{\partial p_0(t)}\geq 0$ and thus $a_{1,i}a_{3,i}-a_{2,i}\geq 0$. Besides, the harvested energy should be no less than 0 for all $p_0(t)\geq 0$. We have that $\lim_{p_0(t)\!\to+\!\infty}e_{i,\rm h}(t)=a_{1,i}-\frac{a_{2,i}}{a_{3,i}}=\frac{a_{1,i}a_{3,i}-a_{2,i}}{a_{3,i}}\geq 0$ and thus $a_{3,i}>0$. By substituting $a_{1,i}a_{3,i}-a_{2,i}\geq 0$ and $a_{3,i}>0$ into $\frac{\partial^2 e_{i,\rm h}(t)}{\partial p_0(t)^2}$, we have that $\frac{\partial^2 e_{i,\rm h}(t)}{\partial p_0(t)^2}\leq 0$. Therefore, $e_{i,\rm h}(t)$ is a concave function of $p_0(t)$.

\section{Proof of Lemma \ref{lem6}}\label{app0}
Based on the update rule of $Q_i(t)$, we have that 
\begin{equation}\label{lem6_eq1}
\small
\begin{split}
\frac{1}{2}\left(Q_i(t+1)\right)^2 \!-\! &\frac{1}{2}\left(Q_i(t)\right)^2\! = \! \frac{1}{2}\left(r_i(t)\!-\!D_{i,\rm C}(t)\!-\!D_{i,\rm O}(t)\right)^2\!\! + \!Q_i(t)\left(\!r_i(t)\!-\!D_{i,\rm C}(t)\!-\!D_{i,\rm O}(t)\right)\\
&\leq \frac{1}{2}\left[\left(D_{i,\rm C}(t)\!+\!D_{i,\rm O}(t)\right)^2\!+\!\left(r_i(t)\right)^2\right]\!-\!Q_i(t)\left(D_{i,\rm C}(t)+D_{i,\rm O}(t)\!-\!r_i(t)\right).
\end{split}
\end{equation}
By summing over the $i=1,\cdots,K$, and taking conditional expectation of \eqref{lem6_eq1}, we have 
\begin{equation}\label{lem6_eq2}
\small
\begin{split}
\mathbb{E}\left[\!\frac{1}{2}\mathsmaller{\sum_{i=1}^{K}}\!\!\left[\left(Q_i(t+1)\right)^2\!\! -\! \left(Q_i(t)\right)^2\right]\!\!\mid\!\!\Theta_t\!\right]\!\!\leq\!\! \mathsmaller{\sum_{i=1}^{K}}\!C_{1,i}\!-\!\mathbb{E}\left[\!\mathsmaller{\sum_{i=1}^{K}}\!Q_i(t)\left(D_{i,\rm C}(t)\!+\!D_{i,\rm O}(t)\!-\!r_i(t)\right)\!\!\mid\!\!\Theta_t\!\right].
\end{split}
\end{equation}
Here, $C_{1,i}$ is a constant obtained as 
\begin{equation}\label{lem6_eq3}
\small
\mathbb{E}\left[\frac{1}{2}\left[\left(D_{i,\rm C}(t)+D_{i,\rm O}(t)\right)^2+\left(r_i(t)\right)^2\right]\right]\leq \frac{1}{2}\left[\left(D_{i,\rm C}^{\rm max}+D_{i,\rm O}^{\rm max}\right)^2+\left(r_{\rm max}\right)^2\right]\triangleq C_{1,i},
\end{equation}
where $D_{i,\rm O}^{\rm max}\!\!=\!\!\mathbb{E}\left[WT\log_2\left(1\!+\!p_{\rm max}\gamma_i(t)\right)\right]$ and  $D_{i,\rm C}^{\rm max}\!\!=\!\!\frac{f_i^{\rm max}T}{\phi_i}$. Following the similar steps, we have
\begin{equation}\label{lem6_eq4}
\small
\begin{split}
\mathbb{E}\left[\frac{1}{2}\left(Q_0(t+1)\right)^2 - \frac{1}{2}\left(Q_0(t)\right)^2\mid\Theta_t\right] &\leq C_2 - \mathbb{E}\left[Q_0(t)\left(D_{0,\rm C}(t)\!-\!\mathsmaller{\sum_{i=1}^{k}}D_{i,\rm O}(t)\right)\mid\Theta_t\right],
\end{split}
\end{equation} 
\begin{equation}\label{lem6_eq6}
\small
\mathbb{E}\left[\frac{1}{2}\mathsmaller{\sum_{i=1}^{K}}\left(\tilde{B}_{t+1}^2-\tilde{B}_t^2\right)\mid\Theta_t\right]\leq \mathsmaller{\sum_{i=1}^{K}}C_{3,i}\! -\! \mathbb{E}\left[\mathsmaller{\sum_{i=1}^{K}}\lambda_{\rm e}(B_i(t)\!-\!\Omega_i)(e_i(t)\!-\!e_{i,\rm h}(t))\mid\Theta_t\right],
\end{equation}
where $C_2\!=\!\frac{1}{2}\!\left[\!\left(\!\small{\sum_{i=1}^{k}}\!D_{i,\rm O}^{\rm max}\!\right)^2\!\!+\!\!\left(\!D_{0,\rm C}^{\rm max}\right)^2\!\right]$ with $D_{0,\rm C}^{\rm max}\!\!=\!\!\frac{f_0^{\rm max}T}{\phi_0}$, and $C_{3,i}\!\!=\!\!\frac{1}{2}\left[\!(\lambda_{\rm e}e_i^{\rm max})^2\!\!+\!\!\left(\lambda_{\rm e}e_{i,\rm h}^{\rm max}\right)^2\!\right]$ with $e_i^{\rm max}=e_{\rm unit}^{\rm col}r_{\rm max} + p_i^{\rm max}T+\kappa_i\left(f_i^{\rm max}\right)^3T$ and $e_{i,\rm h}^{\rm max}=p_0^{\rm max}T$.

For the virtual queue $Z_0(t)$, we use the inequality $\left[\max(\cdot,0)\right]^2\leq(\cdot)^2$ and obtain that
\begin{equation}\label{lem6_eq7}
\small
\begin{split}
\frac{1}{2}\left(Z_0(t+1)\right)^2-\frac{1}{2}\left(Z_0(t)\right)^2\leq& \frac{1}{2}\left(\lambda_{\rm c} e_0(t)-\lambda_{\rm c} e_0^{\rm th}\right)^2 + Z_0(t)\lambda_{\rm c}\left(e_0(t)-e_0^{\rm th}\right)\\
\leq& \frac{1}{2}\left(\lambda_{\rm c} e_0(t)\right)^2+\frac{1}{2}\left(\lambda_{\rm c} e_0^{\rm th}\right)^2 + Z_0(t)\lambda_{\rm c}\left(e_0(t)-e_0^{\rm th}\right).
\end{split}
\end{equation}
Correspondingly, we have 
\begin{equation}\label{lem6_eq8}
\small
\mathbb{E}\left[\frac{1}{2}\left(Z_0(t+1)\right)^2-\frac{1}{2}\left(Z_0(t)\right)\mid\Theta_t\right]\leq C_4 - \mathbb{E}\left[Z_0(t)\lambda_{\rm c}\left(e_0^{\rm th}-e_0(t)\right)\mid\Theta_t\right],
\end{equation}
where $C_4=\frac{1}{2}\left[\left(\lambda_{\rm c}e_0^{\rm max}\right)^2+\left(\lambda_{\rm c}e_0^{\rm th}\right)^2\right]$ and $e_0^{\rm max}=\kappa_0\left(f_0^{\rm max}\right)^3T + p_0^{\rm max}T$.

By substituting \eqref{lem6_eq2}, \eqref{lem6_eq4}, \eqref{lem6_eq6} and \eqref{lem6_eq8} into \eqref{LyaDrift}, we obtain the upper bound in \eqref{LyaDriPlusPen}.

\vspace{-1em}
\section{Proof of Proposition \ref{lem_opt_solution}}\label{App_Opt_solution}
We solve \eqref{Prob_Per_slot_sim_sub} by considering following three cases in terms of the value of $\tau_i$ and $Q_i-Q_0$.


\textbf{Case I}: $\tau_i>0$ and $Q_i-Q_0\geq0$. Notice that $\tilde{B}_i=B_i-\Omega_i \leq 0$ always holds. In this case, $F\left(f_i\right)$ and $G\left(p_i\right)$ are concave functions and achieve maximum at $\tilde{f}_i=\sqrt{\frac{-Q_i}{3\lambda_{\rm e}\tilde{B}_i\kappa_i\phi_i}}$ and $\tilde{p}_i=\frac{\left(Q_0-Q_i\right)W}{\lambda_{\rm e}\tilde{B}_i\ln2}-\frac{1}{\gamma_i}$, respectively. When $\tilde{B}_i = 0$, we set $\tilde{f}_i = f_i^{\rm max}$ and $\tilde{p}_i=p_i^{\rm max}$. Denote $\hat{f}_i=\min\left(\tilde{f}_i, \bar{f}_i^{\rm th}\right)$ and $\hat{p}_i=\left[\tilde{p}_i\right]^{\bar{p}_i^{\rm th}}_0$. Then, the optimal solution must be in the region of $f_i^{\ast}\in\left[0,\hat{f}_i\right]$ and $p_i^{\ast}\in\left[0,\hat{p}_i\right]$, where both $F\left(f_i\right)$ and $G\left(p_i\right)$ are monotonically increasing. Let $\hat{D}_{i,\rm C}^{\rm max}=\frac{\hat{f}_iT}{\phi_i}$ and $\hat{D}_{i,\rm O}^{\rm max}=W\tau_i\log_2\left(1+\hat{p}_i\gamma_i\right)$ denote the maximum amount of data processed via local computing and computation offloading in time slot $t$, respectively. We derive the optimal solution as below:
	
	1) When $\hat{D}_{i,\rm O}^{\rm max}+\hat{D}_{i,\rm C}^{\rm max} \leq Q_i$, we can directly have that $f_i^{\ast}=\hat{f}_i$ and $p_i^{\ast}=\hat{p}_i$. 
	
	2) When $\hat{D}_{i,\rm O}^{\rm max}+\hat{D}_{i,\rm C}^{\rm max} > Q_i$, ${D}_{i,\rm O}+{D}_{i,\rm C} = Q_i$ must hold at optimum. By substituting $p_i=\mathcal{F}_{p_i}(f_i)$ into $G(p_i)$, we can equivalently express \eqref{Prob_Per_slot_sim_sub} as
		\begin{equation}\label{Prob_Per_slot_sim_case3}
			\small
			\underset{\substack{f_i}} \max~ U\left(f_i\right),~~\st
			~~f_i^{\rm lb} \leq\! f_i \!\leq\! f_i^{\rm ub},
		\end{equation}
		where $U\left(f_i\right)=\lambda_{\rm e}\tilde{B}_i\left[\mathcal{F}_{p_i}\left(f_i\right)\tau_i+\kappa_i\left(f_i\right)^3T\right]\!+\!\left(Q_i\right)^2 + \frac{Q_0f_iT}{\phi_i} - Q_0Q_i$, $f_i^{\rm ub} = \hat{f}_i$, and $f_i^{\rm lb} = \max\left(0, \mathcal{F}_{f_i}\left(\hat{p}_i\right)\right)$. When $\tilde{B}_i= 0$, $U\left(f_i\right)$ is a linear function of $f_i$. The optimal solution of \eqref{Prob_Per_slot_sim_case3} is $f_i^{\ast}=\hat{f}_i$ and thus $p_i^{\ast}=\mathcal{F}_{p_i}\left(\hat{f}_i\right)$. When $\tilde{B}_i< 0$, $U\left(f_i\right)$ is a concave function. Suppose $U\left(f_i\right)$ reaches maximum when $f_i=\bar{f}_i$. We obtain $\bar{f}_i$ by solving equation
		$U^\prime(f_i) = \frac{\partial U}{\partial f_i} = 3\lambda_{\rm e}\kappa_iT\tilde{B}_i\left(f_i\right)^2 - \lambda_{\rm e}\tilde{B}_iTA_02^{-\frac{f_iT}{W\tau_i\phi_i}} + \frac{T}{\phi_i}Q_0 = 0$, where $A_0 = \frac{\ln 2}{W\phi_i\gamma_i}2^{\frac{Q_i}{W\tau_i}}$. By taking the derivative of $U^\prime$ with respect to $f_i$, we can easily find that $\frac{\partial U^\prime}{\partial f_i}<0$ for $f_i\in[0,+\infty)$. That is, $U^\prime(f_i)$ is a monotonically decreasing function of $f_i$. Besides, we have $U^\prime(f_i)=-\lambda_{\rm e}TA_0\tilde{B}_i + \frac{T}{\phi_i}Q_0>0$ when $f_i=0$, and $U'(f_i)\to -\infty$ When $f_i\to +\infty$. Therefore, $U^\prime(f_i)=0$ has a unique solution $\bar{f}_i\in[0,+\infty)$. As a result, we can efficiently obtain $\bar{f}_i$ via bi-section search. Then, for the case of $\tilde{B}_t< 0$, the optimal solution of \eqref{Prob_Per_slot_sim_case3} is $f_i^{\ast}=\breve{f}_i$ and thus $p_i^{\ast}=\mathcal{F}_{p_i}\left(\breve{f}_i\right)$, where $\breve{f}_i=\min\left(\max\left(f_i^{\rm lb},\bar{f}_i\right),f_i^{\rm ub}\right)$.

	\textbf{Case II}: $\tau_i>0$ and $Q_i-Q_0<0$. In this case, $F$ is a concave function of $f_i$, while $G$ monotonically decreases with $p_i$. Then, the optimal solution of \eqref{Prob_Per_slot_sim_sub} is $p_i^{\ast}=0$ and $f_i^{\ast}=\hat{f}_i$.
	
	\textbf{Case III}: $\tau_i=0$. In this case, the objective of \eqref{Prob_Per_slot_sim_sub} becomes $F(f_i)$ which is a concave function of $f_i$. The optimal solution can be easily obtained as $p_i^{\ast}=0$ and $f_i^{\ast}=\hat{f}_i$.
	
	By summarizing the three cases above, we finally obtain the result in \eqref{Opt_pu_fu}.

\vspace{-1em}
\section{Proof of Proposition \ref{lem3}}\label{app1}
We initially set $\Omega_i\geq \lambda_{\rm e}e_i^{\rm max}+\lambda_{\rm e}p_0^{\rm max}T$, $\forall i$. In the following, we seek a threshold for $\Omega_i$ so that the constraint \eqref{EH_caus} is satisfied for all $B_i(t)\in[0,\Omega_i]$, $\forall i$. In particular, we consider the following three cases. 

\textbf{Case I}: When $B_i(t)\in[\lambda_{\rm e}e_i^{\rm max}, \Omega_i]$, we have $B_i(t+1) \leq \min\left(\Omega_i+\lambda_{\rm e}p_0^{\rm max}T,\Omega_i\right) = \Omega_i$ based on the update rule of $B_i(t)$ in \eqref{E_evol}. Since $B_i(t)\geq \lambda_{\rm e}e_i^{\rm max}\geq \lambda_{\rm e}e_i(t)$ for all feasible $r_i(t)$, $\tau_i(t)$, $f_i(t)$ and $p_i(t)$, the energy causality constraint \eqref{EH_caus} is satisfied, and thus $0\leq B_i(t+1)\leq \Omega_i$.
	
\textbf{Case II}: When $B_i(t)\in[0, B_{\rm min}]$, we have $B_i(t+1) \leq B_i(t) + \lambda_{\rm e}e_{i,\rm h}(t) \leq B_{\rm min} + \lambda_{\rm e}p_0^{\rm max}T< \Omega_i$. In this case, we have $e_{i,\rm col}(t)=e_{i,\rm O}(t)=e_{i,\rm C}(t)=0$ under the energy-aware management policy. The energy causality constraint \eqref{EH_caus} is obviously satisfied and thus $0\leq B_i(t+1)< \Omega_i$.
	
\textbf{Case III}: When $B_i(t)\in[B_{\rm min}, \lambda_{\rm e}e_i^{\rm max}]$, we have $B_i(t+1) \leq \lambda_{\rm e}e_i^{\rm max}+\lambda_{\rm e}p_0^{\rm max}T\leq \Omega_i$. To respect the energy causality $\lambda_{\rm e}e_i(t)\leq B_i(t)$ for all $B_i(t)\in[B_{\rm min}, \lambda_{\rm e}e_i^{\rm max}]$, one possible solution is to spend zero-Joule energy on data sensing and $B_{\rm min}$-Joule energy at most on task offloading and local computing. Accordingly, we determine the threshold of $\Omega_i$ in this case as follows.

	1) Based on \eqref{Opt_r}, we have that $r_i(t)=0$ when $\frac{V-Q_i(t)}{\lambda_{\rm e}e_{\rm unit}^{\rm col}}+B_i(t)<\Omega_i$. Since $\frac{V-Q_i(t)}{\lambda_{\rm e}e_{\rm unit}^{\rm col}}+B_i(t)<\frac{V}{\lambda_{\rm e}e_{\rm unit}^{\rm col}}+\lambda_{\rm e}e_i^{\rm max}$ for $B_i(t)\in[B_{\rm min}, \lambda_{\rm e}e_i^{\rm max}]$,  where the inequality holds by dropping the negative terms $\frac{-Q_i(t)}{\lambda_{\rm e}e_{\rm unit}^{\rm col}}$ and setting $B_i(t)=\lambda_{\rm e}e_i^{\rm max}$. Then, we can set $\Omega_i\geq \frac{V}{\lambda_{\rm e}e_{\rm unit}^{\rm col}}+\lambda_{\rm e}e_i^{\rm max}$ such that the energy cost on data sensing is zero.
		
	2) Based on \eqref{Opt_pu_fu} and \eqref{Opt_t}, the WD$_i$ consumes most energy on data transmission and local computing when $\left\{f_i^{\ast}(t),p_i^{\ast}(t)\right\} = \left\{\hat{f}_i(t), \hat{p}_i(t)\right\}$ and $\tau_i(t)=T$. To ensure $\lambda_{\rm e}\left(e_{i,\rm O}(t)+e_{i,\rm C}(t)\right)\leq B_{\rm min}$, it is sufficient to satisfy $\lambda_{\rm e}\left(\hat{e}_{i,\rm O}(t)+\hat{e}_{i,\rm C}(t)\right)\leq B_{\rm min}$, where $\hat{e}_{i,\rm O}(t) = \hat{p}_i(t)T$ and $\hat{e}_{i,\rm C}(t) = \kappa_{\rm e}\left(\hat{f}_i(t)\right)^3T$. Recall that $\hat{p}_i(t) = \left[\tilde{p}_i(t)\right]^{\bar{p}_i^{\rm th}(t)}_0$ and $\hat{f}_i(t)=\min\left(\tilde{f}_i(t),\bar{f}_i^{\rm th}(t)\right)$, where $\tilde{f}_i(t)=\sqrt{\frac{-Q_i(t)}{3\lambda_{\rm e}\tilde{B}_i(t)\kappa_i\phi_i}}$ and $\tilde{p}_i(t)=\frac{\left(Q_0(t)-Q_i(t)\right)W}{\lambda_{\rm e}\tilde{B}_i(t)\ln2}-\frac{1}{\gamma_i(t)}$, respectively. Obviously, $\hat{f}_i(t)\leq \tilde{f}_i(t)$. In the following, we determine the threshold of $\Omega_i$ in two sub-cases: 
		
		a) When $\tilde{p}_i(t)> 0$, we have $\hat{p}_i(t)\leq \tilde{p}_i(t)$. Since $Q_i(t)\leq V+r_{\rm max}$ (see Lemma \ref{lem_ub}), we have
		\vspace{5pt} 
		\begin{equation}\label{App1_eq1}
			\small
			\hat{f}_i(t)\!\! \leq\!\! \sqrt{\frac{-Q_i(t)}{3\lambda_{\rm e}\tilde{B}_i(t)\kappa_i\phi_i}} \!\!\leq\!\! \sqrt{\frac{V\!+\!r_{\rm max}}{3\lambda_{\rm e}\left(\Omega_i\!-\!\lambda_{\rm e}e_i^{\rm max}\right)\kappa_i\phi_i}},
			\hat{p}_i(t)\! \leq \!\frac{\left(Q_0(t)\!\!-\!\!Q_i(t)\right)W}{\lambda_{\rm e}\tilde{B}_i(t)\ln2}\!-\!\frac{1}{\gamma_i(t)}\!\!\leq\!\! \frac{\left(V\!+\!r_{\rm max}\right)W}{\lambda_{\rm e}\left(\Omega_i\!-\!\lambda_{\rm e}e_i^{\rm max}\right)\ln 2}.\vspace{5pt} 
		\end{equation}
		By submitting \eqref{App1_eq1} into $\lambda_{\rm e}\left(\hat{e}_{i,\rm O}(t)+\hat{e}_{i,\rm C}(t)\right)\leq B_{\rm min}$, we obtain the threshold of $\Omega_i$ by solving
		\vspace{5pt} 
		\begin{equation}\label{lem4_case3_eq1}
		\small
		\begin{split}
		H(\Omega_i) = \lambda_{\rm e}\kappa_i\left(\sqrt{\frac{V+r_{\rm max}}{3\lambda_{\rm e}\left(\Omega_i-\lambda_{\rm e}e_i^{\rm max}\right)\kappa_i\phi_i}}\right)^3T \!+\! \frac{\left(V+r_{\rm max}\right)W}{\lambda_{\rm e}\left(\Omega_i-\lambda_{\rm e}e_i^{\rm max}\right)\ln 2}T\!-\!B_{\rm min} \leq 0.	
		\end{split}
		\end{equation}
		Notice that $\Omega_i\!\geq\! \lambda_{\rm e}e_i^{\rm max}$. Obviously, $H(\Omega_i)$ monotonically decreases with $\Omega_i\in\left(\lambda_{\rm e}e_i^{\rm max},\!+\!\infty\right)$. Meanwhile, $H(\Omega_i)\to +\infty$ when $\Omega_i\to \lambda_{\rm e}e_i^{\rm max}$ and $H(\Omega_i)\to -B_{\rm min}$ when $\Omega_i\to +\infty$. Therefore, $H(\Omega_i) =0$ has a unique solution $\bar{\Omega_i}$, which can be easily obtained via bisection method. Then, we can satisfy $\lambda_{\rm e}\left(\hat{e}_{i,\rm O}(t)+\hat{e}_{i,\rm C}(t)\right)\leq B_{\rm min}$ by setting $\Omega_i\geq \bar{\Omega_i}$.
		
		b) When $\tilde{p}_i(t)\leq 0$, we have $\hat{p}_i(t)=0$. In this sub-case, we determine the threshold of $\Omega_i$ satisfying $\lambda_{\rm e}\hat{e}_{i,\rm C}(t)\leq B_{\rm min}$ by solving the following inequality:
		\begin{equation}\label{APP1_eq3}
		\small
		\lambda_{\rm e}\kappa_i\left(\sqrt{\frac{V+r_{\rm max}}{3\lambda_{\rm e}\left(\Omega_i-\lambda_{\rm e}e_i^{\rm max}\right)\kappa_i\phi_i}}\right)^3T \leq B_{\rm min}.
		\end{equation}
		Obviously, $\Omega_i\geq \bar{\Omega_i}$ also satisfies \eqref{APP1_eq3}.
In summarize, we can set $\Omega_i\!\geq\!\max\left(\!\frac{V}{\lambda_{\rm e}e_{\rm unit}^{\rm col}}+\lambda_{\rm e}e_i^{\rm max},\bar{\Omega_i}\!\right)+\lambda_{\rm e}p_0^{\rm max}T$, $\forall i$, to make the energy causality constraint \eqref{EH_caus} implicit for all $B_i(t)\in[0,\Omega_i]$, $\forall i$, which completes the proof.

\vspace{-1em}
\section{Proof of Proposition \ref{lem4}}\label{app2}
We start with two lemmas, which are useful to prove Proposition \ref{lem4}.\vspace{-0.5em}
\begin{Lemma}\label{lem2}
	The maximal utility $\bar{R}_{\rm P2}^\ast$ of problem (P2) can be achieved arbitrarily closely by an $\omega$-only policy. That is, for any $\delta>0$, there is an $\omega$-only policy $\Pi$, achieving
	\begin{equation}
	\small
	\mathbb{E}\left[\mathsmaller{\sum_{i=1}^{K}}r_i^{\Pi}(t)\right] \geq \bar{R}_{\rm P2}^\ast - \delta,
	\end{equation}
	while respecting the constraints \eqref{ledge_off_cons}, \eqref{T_allo_cons} and \eqref{Prob_const2} in (P2), and satisfying
	\begin{subequations}\label{AppD_eq2}
		\small
		\setlength{\abovedisplayskip}{5pt}
		\setlength{\belowdisplayskip}{5pt}
		\begin{align}
		&\mathbb{E}\left[e_0^{\Pi}(t)-e_0^{\rm th}\right] \leq \delta,~\mathbb{E}\left[\mathsmaller{\sum_{i=1}^{K}}D_{i,\rm{O}}^{\Pi}(t)\right] \leq \mathbb{E}\left[D_{0,\rm{C}}^{\Pi}(t)\right]+\delta,\\
		&~ \mathbb{E}\left[e_i^{\Pi}(t)-e_{i,\rm h}^{\Pi}(t)\right] \leq \delta,~\mathbb{E}\left[r_i^{\Pi}(t)\right] \leq \mathbb{E}\left[D_{i,\rm{O}}^{\Pi}(t)+D_{i,\rm{C}}^{\Pi}(t)\right] + \delta, \forall i.
		\end{align}
	\end{subequations}
\end{Lemma}
\vspace{-1.5em}
\begin{proof}
	The proof follows similar steps of Theorem 4.5 in \cite{Neely2010} and is omitted here for brevity.
\end{proof}

\vspace{-1em}
\begin{Lemma}\label{lem7}
	If $Z(t)$ is mean rate stable, i.e., $\lim_{N\!\to\!\infty}\!\frac{\mathbb{E}[\!Z(N)\!]}{N}\!=\!0$, then the HAP satisfies the average constraint \eqref{Bud_cons}.
\end{Lemma}
\begin{proof}
	{\color{black}This result is a derivation of Theorem 2.5 in \cite{Neely2010} and can be proved following the similar steps in Section 4.4 in \cite{Neely2010}. The detail proof is omitted here for brevity.}
\end{proof}
\textbf{Proof of Proposition \ref{lem4}:} Denote the policy produced by LEESE as $\Psi$. Since LEESE minimizes the upper bound \eqref{LyaDriPlusPen} on the Lyapunov drift-plus-penalty function $\Delta^{t}_V$, we have that
\begin{small}
\begin{equation}\label{Lem_eq1}
\begin{split}
&\Delta^{t}_V=\mathbb{E}\left[\Phi(t+1)-\Phi(t)|\Theta_t\right]-V\mathbb{E}\left[\mathsmaller{\sum_{i=1}^{K}}r_i(t)|\Theta(t)\right]\\
&\leq C + \upsilon -\! \mathbb{E}\left\{L^\Psi(t)\mid\!\!\Theta_t\right\}\leq C + \upsilon-\! \mathbb{E}\left\{L^\Pi(t)\mid\!\!\Theta_t\right\}\\
&\overset{(\dag)}= C\!+ \!\upsilon \!-\! \mathbb{E}\left[\mathsmaller{V\sum_{i=1}^{K}}r_i^\Pi(t)\right]\! +\! \mathbb{E}\left[\mathsmaller{\sum_{i=1}^{K}}Q_i(t)\left(r_i^\Pi(t)\!-\!D_{i,\rm C}^\Pi(t)\!-\!D_{i,\rm O}^\Pi(t)\right)\right]+ \mathbb{E}\left[Z_0(t)\lambda_{\rm c}\left(e_0^\Pi(t)-e_0^{\rm th}\right)\right]\\
& + \mathbb{E}\left[Q_0(t)\left(\mathsmaller{\sum_{i=1}^{k}}D_{i,\rm O}^\Pi(t)-D_{0,\rm C}^\Pi(t)\right)\right]+\mathbb{E}\left[\mathsmaller{\sum_{i=1}^{K}}\lambda_{\rm e}(\Omega_i-B_i(t))(e_i^\Pi(t)-e_{i,\rm h}^\Pi(t))\right] \\
&\overset{(\ddag)}\leq C + \upsilon - V\left(\bar{R}_{\rm P2}^\ast-\delta\right) + \left[\mathsmaller{\sum_{i=1}^{K}}Q^{\rm max} + Q_0(t) + \mathsmaller{\sum_{i=1}^{K}}\lambda_{\rm e}(\Omega_i) + Z_0(t)\lambda_{\rm c}\right]\delta,
\end{split}
\end{equation}
\end{small}where $L^\Psi(t)$ and $L^\Pi(t)$ are the value of $L(t)$ in Lemma \ref{lem1} under policy $\Psi$ and $\Pi$, respectively. $(\dag)$ holds due to the independence of policy $\Pi$ on $\Theta_t$. $(\ddag)$ is obtained by plugging \eqref{AppD_eq2} and using the fact that $Q_i(t)\leq Q^{\rm max}$ and $B_i(t)\leq \Omega_i$, $\forall i$. Let $\delta \to 0$, we have that
\begin{equation}\label{Eq1_app2}
\small
\mathbb{E}\left[\Phi(t+1)-\Phi(t)|\Theta_t\right]-V\mathbb{E}\left[\mathsmaller{\sum_{i=1}^{K}}r_i^{\Psi}(t)|\Theta(t)\right]\leq C+ \upsilon -V\bar{R}_{\rm P2}^\ast.
\end{equation}
By summing up the both sides of \eqref{Eq1_app2} over $t=0,\cdots,N-1$, taking iterated expectations, and then dividing by $VN$, we have that 
\begin{equation}
\small
\frac{\mathbb{E}\left[\Phi(N)\right]\!-\!\mathbb{E}\left[\Phi(0)\right]}{NV}\!-\!\frac{1}{N}\mathsmaller{\sum}_{t=0}^{N-1}\mathbb{E}\left[\mathsmaller{\sum_{i=1}^{K}}r_i^{\Psi}(t)\right]\!\leq\! \frac{C+ \upsilon}{V}\!-\!\bar{R}_{\rm P2}^\ast.
\end{equation}
By rearranging terms and setting $N \to \infty$, we prove a) that 
\begin{equation}
\small
\bar{R}_{\Psi}= \lim_{N\!\to+\!\infty}\!\!\frac{1}{N}\!\mathsmaller{\sum}_{t=0}^{N\!-\!1}\mathbb{E}\left[\mathsmaller{\sum_{i=1}^{K}}r_i^{\Psi}(t)\right]\!\!\geq \!\bar{R}_{\rm P2}^\ast\!-\!\frac{C+ \upsilon}{V} \!\overset{(\S)}\geq\! \bar{R}_{\rm P1}^\ast\!-\!\frac{C+ \upsilon}{V},
\end{equation}
where $\S$ holds for $\bar{R}_{\rm P1}^\ast \leq \bar{R}_{\rm P2}^\ast$.

To prove b) and c), we plug the $\omega$-only policy $\Gamma$ that satisfies the Slater conditions \eqref{SLT} into the RHS of the inequality $(\dag)$ in \eqref{Lem_eq1}. By dropping the negative term $-\epsilon\sum_{i=1}^{K}\lambda_{\rm e}\left(\Omega_i-B_i(t)\right)$, we obtain
\begin{equation}
\small
\mathbb{E}\left[\!\Phi(t+1)\!-\!\Phi(t)|\Theta_t\!\right]-\!V\mathbb{E}\left[\!{\mathsmaller\sum_{i=1}^{K}}\!r_i^{\Psi}(t)\!|\Theta(t)\!\right]\!\leq\! C\!+\! \upsilon-\!V\varphi(\epsilon)\!-\!\left(\!Z_0(t)\lambda_{\rm c}\!+\!{\mathsmaller\sum_{i=1}^{K}}\!Q_i(t)\!+\!Q_0(t)\!\right)\epsilon.
\end{equation}
{\color{black}Taking iterated expectations and telescoping sums over $t=0,\cdots,N-1$, normalizing by $N\epsilon$, and letting $N\to\infty$, we have that
\begin{equation}\label{app2_eq1}
\small
\lim_{N\!\to\!+\!\infty}\!\frac{1}{N}\!\mathsmaller\sum_{t=0}^{N\!-\!1}\!\mathbb{E}\!\left[Z_0(t)\lambda_{\rm c}\!+\!Q_0(t)\!+\!\!{\mathsmaller\sum_{i=1}^{K}}Q_i(t)\right] \!\!\leq \!\!\frac{C\!\!+\!\! \upsilon\!\!+\!\!V\left(\bar{R}_{\Psi}\!\!-\!\!\varphi(\epsilon)\right)}{\epsilon}\!\!\overset{(\S)}\leq\!\frac{C\!+\! \upsilon\!\!+\!\!V\left(\bar{R}_{\rm P2}^\ast\!\!-\!\!\varphi(\epsilon)\right)}{\epsilon},
\end{equation}
where $\S$ uses the fact that $\bar{R}_{\Psi}\leq \bar{R}_{\rm P2}^\ast$. \eqref{app2_eq1} implies that $Z_0(t)$, $Q_0(t)$ and $Q_i(t)$ are strongly stable, i.e., 
\begin{equation}\label{app2_eq2}
\small
\lim_{N\!\to\!+\!\infty}\frac{1}{N}\mathsmaller{\sum}_{t=0}^{N-1}\mathbb{E}\left[Z_0(t)\right]\!<\!\infty,\lim_{N\!\to\!+\!\infty}\frac{1}{N}\mathsmaller{\sum}_{t=0}^{N-1}\mathbb{E}\left[Q_0(t)\right]\!<\!\infty,\lim_{N\!\to\!+\!\infty}\frac{1}{N}\mathsmaller{\sum}_{t=0}^{N-1}\mathbb{E}\left[Q_i(t)\right]\!<\!\infty.
\end{equation}
Because a strongly stable $Z_0(t)$ is also mean rate stable (see Theorem 2.8 in \cite{Neely2010}), the long-term average power constraint \eqref{Bud_cons} is satisfied according to Lemma \ref{lem7}, which proves b). 

Besides, as shown in \eqref{app2_eq2}, LEESE satisfies all the long-term constraints in (P1). When the battery capacity of each WD is larger than the threshold in Proposition \ref{lem3}, LEESE produces a feasible solution to (P1). In this case, we always have $\bar{R}_{\Psi}\leq \bar{R}_{\rm P1}^\ast$. By substituting this into the RHS of $\S$ in \eqref{app2_eq1} and using the fact that $Z_0(t)$, $Q_0(t)$ and $Q_i(t)$ are non-negative, we obtain \eqref{lem5_eq2} and thus prove c).}

\vspace{-1em}
\ifCLASSOPTIONcaptionsoff
  \newpage
\fi

\bibliographystyle{IEEEtran}
\bibliography{IEEEabrv,MyRefLib}

\begin{thebibliography}{10}
\providecommand{\url}[1]{#1}
\csname url@samestyle\endcsname
\providecommand{\newblock}{\relax}
\providecommand{\bibinfo}[2]{#2}
\providecommand{\BIBentrySTDinterwordspacing}{\spaceskip=0pt\relax}
\providecommand{\BIBentryALTinterwordstretchfactor}{4}
\providecommand{\BIBentryALTinterwordspacing}{\spaceskip=\fontdimen2\font plus
\BIBentryALTinterwordstretchfactor\fontdimen3\font minus
  \fontdimen4\font\relax}
\providecommand{\BIBforeignlanguage}[2]{{%
\expandafter\ifx\csname l@#1\endcsname\relax
\typeout{** WARNING: IEEEtran.bst: No hyphenation pattern has been}%
\typeout{** loaded for the language `#1'. Using the pattern for}%
\typeout{** the default language instead.}%
\else
\language=\csname l@#1\endcsname
\fi
#2}}
\providecommand{\BIBdecl}{\relax}
\BIBdecl

\bibitem{Ju2014ThroughputMaximization}
H.~Ju and R.~Zhang, ``Throughput maximization in wireless powered communication
  networks,'' \emph{{IEEE} Trans. Wireless Commun.}, vol.~13, no.~1, pp.
  418--428, Jan. 2014.

\bibitem{Bi2015}
S.~Bi, C.~K. Ho, and R.~Zhang, ``Wireless powered communication: opportunities
  and challenges,'' \emph{IEEE Commun. Mag.}, vol.~53, no.~4, pp. 117--125,
  Apr. 2015.

\bibitem{Zeng2017}
Y.~Zeng, B.~Clerckx, and R.~Zhang, ``Communications and signals design for
  wireless power transmission,'' \emph{IEEE Trans. Commun.}, vol.~65, no.~5,
  pp. 2264--2290, May 2017.

\bibitem{Li2019}
X.~Li, X.~Zhou, C.~Sun, and D.~W.~K. Ng, ``Online policies for throughput
  maximization of energy-constrained wireless-powered communication systems,''
  \emph{IEEE Trans. Wireless Commun.}, vol.~18, no.~3, pp. 1463--1476, Jan.
  2019.

\bibitem{Mao2017}
Y.~{Mao}, C.~{You}, J.~{Zhang}, K.~{Huang}, and K.~B. {Letaief}, ``A survey on
  mobile edge computing: the communication perspective,'' \emph{IEEE Commun.
  Surv. Tutor.}, vol.~19, no.~4, pp. 2322--2358, Aug. 2017.

\bibitem{Wu2018TVT}
Y.~Wu, K.~Ni, C.~Zhang, L.~P. Qian, and D.~H.~K. Tsang, ``{NOMA}-assisted
  multi-access mobile edge computing: A joint optimization of computation
  offloading and time allocation,'' \emph{{IEEE} Trans. Veh. Technol.},
  vol.~67, no.~12, pp. 12\,244--12\,258, Dec. 2018.

\bibitem{Bi2021}
\BIBentryALTinterwordspacing
S.~Bi, L.~Huang, H.~Wang, and Y.-J.~A. Zhang, ``Lyapunov-guided deep
  reinforcement learning for stable online computation offloading in
  mobile-edge computing networks,'' \emph{IEEE Trans. Wireless Commun.}, 2021.
  [Online]. Available: \url{10.1109/TWC.2021.3085319.}
\BIBentrySTDinterwordspacing

\bibitem{Wang2016}
Y.~{Wang}, M.~{Sheng}, X.~{Wang}, L.~{Wang}, and J.~{Li}, ``Mobile-edge
  computing: partial computation offloading using dynamic voltage scaling,''
  \emph{IEEE Trans. Commun.}, vol.~64, no.~10, pp. 4268--4282, Oct. 2016.

\bibitem{Wang2018a}
F.~{Wang}, J.~{Xu}, X.~{Wang}, and S.~{Cui}, ``Joint offloading and computing
  optimization in wireless powered mobile-edge computing systems,'' \emph{IEEE
  Trans. Wireless Commun.}, vol.~17, no.~3, pp. 1784--1797, Mar. 2018.

\bibitem{Hu2018}
X.~Hu, K.-K. Wong, and K.~Yang, ``Wireless powered cooperation-assisted mobile
  edge computing,'' \emph{IEEE Trans. Wireless Commun.}, vol.~17, no.~4, pp.
  2375--2388, Jan. 2018.

\bibitem{Zhou2020}
F.~Zhou and R.~Q. Hu, ``Computation efficiency maximization in wireless-powered
  mobile edge computing networks,'' \emph{IEEE Trans. Wireless Commun.},
  vol.~19, no.~5, pp. 3170--3184, May 2020.

\bibitem{Mao2016}
Y.~Mao, J.~Zhang, and K.~B. Letaief, ``Dynamic computation offloading for
  mobile-edge computing with energy harvesting devices,'' \emph{IEEE J. Sel.
  Areas Commun.}, vol.~34, pp. 3590--3605, Dec. 2016.

\bibitem{Min2019}
M.~Min, L.~Xiao, Y.~Chen, P.~Cheng, D.~Wu, and W.~Zhuang, ``Learning-based
  computation offloading for {IoT} devices with energy harvesting,'' \emph{IEEE
  Trans. Veh. Technol.}, vol.~68, pp. 1930--1941, Feb. 2019.

\bibitem{XIANarxiv2021}
X.~Li, S.~Bi, Z.~Quan, and H.~Wang, ``Online cognitive data sensing and
  processing optimization in energy-harvesting edge computing systems,'' Jun.
  2021, [Online]. Available: http://arxiv.org/abs/2106.14113.

\bibitem{Xu2017a}
J.~Xu, L.~Chen, and S.~Ren, ``Online learning for offloading and autoscaling in
  energy harvesting mobile edge computing,'' \emph{IEEE Trans. Cogn. Commun.
  Netw.}, vol.~3, pp. 361--373, Sep. 2017.

\bibitem{Wei2019}
Z.~Wei, B.~Zhao, J.~Su, and X.~Lu, ``Dynamic edge computation offloading for
  internet of things with energy harvesting: A learning method,'' \emph{IEEE
  Trans. Veh. Technol.}, vol.~6, no.~3, pp. 4436--4447, Jun. 2019.

\bibitem{Wang2020c}
F.~Wang, J.~Xu, and S.~Cui, ``Optimal energy allocation and task offloading
  policy for wireless powered mobile edge computing systems,'' \emph{{IEEE}
  Trans. Wireless Commun.}, vol.~19, no.~4, pp. 2443--2459, Apr. 2020.

\bibitem{Wu2019}
H.~Wu, X.~Lyu, and H.~Tian, ``Online optimization of wireless powered
  mobile-edge computing for heterogeneous industrial internet of things,''
  \emph{IEEE Internet Things J.}, vol.~6, no.~6, pp. 9880--9892, Dec. 2019.

\bibitem{Sun2021}
M.~Sun, X.~Xu, Y.~Huang, Q.~Wu, X.~Tao, and P.~Zhang, ``Resource management for
  computation offloading in {D2D}-aided wireless powered mobile-edge computing
  networks,'' \emph{IEEE Internet Things J.}, vol.~8, no.~10, pp. 8005--8020,
  May 2021.

\bibitem{Deng2020}
\BIBentryALTinterwordspacing
X.~Deng, J.~Li, L.~Shi, Z.~Wei, X.~Zhou, and J.~Yuan, ``Wireless powered mobile
  edge computing: Dynamic resource allocation and throughput maximization,''
  \emph{IEEE Trans. Mob. Comput.}, pp. 1--1, Oct. 2020. [Online]. Available:
  \url{10.1109/TMC.2020.3034479}
\BIBentrySTDinterwordspacing

\bibitem{Chen2017}
Y.~Chen, N.~Zhao, and M.-S. Alouini, ``Wireless energy harvesting using signals
  from multiple fading channels,'' \emph{IEEE Trans. Commun.}, vol.~65, no.~11,
  pp. 5027--5039, Nov. 2017.

\bibitem{Liu2020}
T.~Liu, X.~Qu, and W.~Tan, ``Online optimal control for wireless cooperative
  transmission by ambient {RF} powered sensors,'' \emph{IEEE Trans. Wireless
  Commun.}, vol.~19, pp. 6007--6019, Jun. 2020.

\bibitem{Wang2018}
F.~Wang, J.~Xu, X.~Wang, and S.~Cui, ``Joint offloading and computing
  optimization in wireless powered mobile-edge computing systems,'' \emph{IEEE
  Trans. Wireless Commun.}, vol.~17, no.~3, pp. 1784--1797, 2018.

\bibitem{Neely2010}
M.~J. Neely, \emph{``Stochastic network optimization with application to
  communication and queueing systems''}.\hskip 1em plus 0.5em minus 0.4em\relax
  Morgan \& Claypool, 2010.

\bibitem{Mehrotra1992}
S.~Mehrotra, ``On the implementation of a primal-dual interior point method,''
  \emph{SIAM J. Optim.}, vol.~2, no.~4, pp. 575--601, Nov. 1992.

\bibitem{Muenzel2015}
V.~Muenzel, A.~F. Hollenkamp, A.~I. Bhatt, J.~de~Hoog, M.~Brazil, D.~A. Thomas,
  and I.~Mareels, ``A comparative testing study of commercial 18650-format
  lithium-ion battery cells,'' \emph{J. Electrochem. Soc.}, vol. 162, no.~8,
  pp. A1592--A1600, May 2015.

\bibitem{Diamond2016}
S.~Diamond and S.~Boyd, ``{CVXPY}: {A} {P}ython-embedded modeling language for
  convex optimization,'' \emph{J. Mach. Learn. Res.}, vol.~17, no.~83, pp.
  1--5, Jan. 2016.

\end{thebibliography}

\end{document}